\documentclass[aps,prx,twocolumn,superscriptaddress,nofootinbib,longbibliography]{revtex4-2}
\ifdefined\pdfminorversion
\fi

\usepackage{amsmath,amssymb,amsthm,booktabs,graphicx,array,placeins}
\usepackage[hidelinks,hypertexnames=false]{hyperref}
\newcommand{\argmax}{\operatorname*{arg\,max}}
\newcommand{\MLD}{\mathrm{MLD}}
\newcommand{\MP}{\mathrm{MP}}
\newtheorem{smtheorem}{Theorem}
\newtheorem{smproposition}[smtheorem]{Proposition}
\newtheorem{smlemma}[smtheorem]{Lemma}
\newtheorem{smdefinition}[smtheorem]{Definition}

\begin{document}

\title{Divergence between Maximum-Probability and Maximum-Likelihood Decoding in the Low-Noise Limit of Quantum Error Correction}

\author{Yixin Zhao}
\email{zhaoyixin22@mails.ucas.ac.cn}
\affiliation{Beijing Key Laboratory of Fault-Tolerant Quantum Computing,
Beijing Academy of Quantum Information Sciences, Beijing 100193, China}
\affiliation{Beijing National Laboratory for Condensed Matter Physics,
Institute of Physics, Chinese Academy of Sciences, Beijing 100190, China}
\affiliation{University of Chinese Academy of Sciences,
Beijing 100190, China}

\author{Fei Yan}
\email{yanfei@baqis.ac.cn}
\affiliation{Beijing Key Laboratory of Fault-Tolerant Quantum Computing,
Beijing Academy of Quantum Information Sciences, Beijing 100193, China}

\date{\today}

\begin{abstract}
In the asymptotic low-noise limit of quantum error correction, two optimal
decoding rules can disagree. Maximum-probability (MP) decoding selects the
most probable microscopic error, whereas degenerate maximum-likelihood (MLD)
decoding selects the logical sector with the largest total probability.
Under code-capacity Pauli noise, we determine the first physical error
weight $m$ at which the MP and MLD logical winner sets become provably
disjoint, even under optimal MP tie resolution. As the physical error rate tends to zero, MLD retains a residual zero-temperature degeneracy from the multiplicity of minimum-weight errors sharing a syndrome and logical label, which can outweigh the single most probable error. Code distance imposes
the universal bound $m\geq h=\lceil d/2\rceil$; the rigidity depth $r$,
defined by $m=h+r$, counts the additional error-weight shells above $h$ on
which MP and MLD remain equivalent. We certify a three-level hierarchy:
$r=0$ for planar surface codes and two recursively concatenated families,
$r=1$ for odd-distance square toric codes and the Gross $[[144,12,12]]$
quantum low-density-parity-check code, and $r=2$ for a separable family
with hypergraph product and bivariate bicycle descriptions. The onset
weight fixes the leading failure-rate gap between best-tie MP and MLD,
proportional to the $m$th power of the physical noise strength, providing
an exact subthreshold benchmark for the advantage of degeneracy-aware
decoding over the computationally cheaper nondegenerate rule.
\end{abstract}

\maketitle

\section{Introduction}

Configurational entropy can overturn a ground-state preference in a finite
disordered system. In quantum error correction, this competition occurs
within the ensemble of physical errors compatible with each measured
syndrome. Standard statistical-mechanical mappings place Pauli
decoding on the Nishimori line, the manifold on which the couplings of a
disordered spin model are matched to the physical noise
\cite{DennisEtAl2002,WangHarringtonPreskill2003,ChubbFlammia2019}.
A measured syndrome partitions the compatible microscopic errors into
logical sectors. Maximum-probability (MP) decoding selects the most probable
microscopic configuration and therefore implements a zero-temperature
ground-state rule. Degenerate maximum-likelihood decoding (MLD) instead
selects the logical sector with the largest partition function and therefore
implements the corresponding finite-temperature free energy rule. A logical
sector can consequently lose through microscopic energy while winning through
configurational multiplicity.

Two complementary limits govern this competition. The threshold limit
takes $d\to\infty$ at fixed noise strength to locate the critical point of
the disordered spin model, where the MP--MLD distinction appears as the
difference between the zero-temperature and Nishimori-line critical
disorder strengths in a re-entrant phase diagram
\cite{WangHarringtonPreskill2003,ChubbFlammia2019,ThomasKatzgraber2011}.
Conversely, the subthreshold limit of a finite-size memory holds the code
fixed and takes the physical noise strength to zero. In this low-noise regime, which we study here, errors organize into
well-separated shells of fixed weight
\cite{SommersEtAl2026,SommersHuseGullans2025}. Optimal decoding then
reduces to counting how many errors of each weight share a syndrome and
logical label, and identifying the first weight shell at which that
multiplicity outweighs the probability of the most probable single error.

Building on this mapping, sum-product recursions, tensor network contractions,
and tensor enumerators provide exact or controlled approximations to logical
partition functions for concatenated, topological, and more general stabilizer
codes
\cite{RahnDohertyMabuchi2002,Poulin2006,FerrisPoulin2014,
CaoLackey2023,CaoEtAl2024}.
Exact and approximate MLD algorithms, together with the complexity boundary
of optimal stabilizer decoding, have also been developed extensively
\cite{DuclosCianciPoulin2010,BravyiSucharaVargo2014,IyerPoulin2015,
deMartiEtAl2024,CaoEtAl2025Exact,CaoEtAl2026Review}.
Within this line of work, minimal multi-goal trellises provide exact max-product
and sum-product
implementations of nondegenerate and degenerate maximum-likelihood decoding,
with constructions that reduce the associated trellis complexity
\cite{StylianouEtAl2025Trellis}.
Together, these methods provide exact or approximate logical sector
probabilities and decoding decisions. A
complementary low-noise literature studies rare logical failures, importance
sampling, parity and boundary effects, finite-size partition functions,
aggregate success ratios, decodability boundaries, and exact finite-noise
behavior
\cite{BravyiVargo2013,Fowler2013,Fowler2013Erratum,BenhemouEtAl2025,
ShuttyEtAl2026,Wichette2026,EnglishEtAl2025,PhysRevA.105.042616,
SommersEtAl2026}.
More recently, circuit-level rare-event methods have introduced decoder
failure spectra, exact counting of minimum-fault failure configurations, and
importance-sampling techniques for specified codes, circuits, and decoders
\cite{BeverlandEtAl2025FailFast}.
Most recently, low-weight errors well below $d/2$ have been identified as
the source of BP-OSD error floors on bivariate bicycle codes
\cite{FregonaEtAl2026}.

These works establish the distinction between the most probable microscopic
error and the most probable logical sector, and provide exact or approximate
methods for evaluating the corresponding decisions. Here we study a different
observable: the minimum physical error weight at which the set-valued MP and
MLD logical winner sets become disjoint, even after optimal MP tie resolution.
Unequal logical sector partition functions alone do not imply a strict
decision split, because MP and MLD may still share a logical winner.
Selecting a single winner from each set using an arbitrary tie convention
can nevertheless make these overlapping winner sets appear disjoint.

For a syndrome $s$ and an interior low-noise Pauli ray
$(p_X,p_Y,p_Z)=t(x,y,z)$, with $x,y,z>0$\footnote{The interior condition
$x,y,z>0$ excludes exactly uniaxial noise: there every weight-$w$ error
carries identical probability, the peak comparison degenerates to a tie
among all sectors with minimum-weight representatives, and best-tie MP
coincides with MLD at leading order. Interior anisotropic directions are
the generic case for biased Pauli noise.}, $x+y+z=1$, and $t\to0$,
let $\Pr_t(E)$ denote the probability of a physical Pauli error $E$.
The two statistical weights are
\begin{equation}
 \begin{aligned}
 T_{s,L}(t)&=\max_{E\in C_sL}\Pr_t(E),\\
 M_{s,L}(t)&=\sum_{E\in C_sL}\Pr_t(E),
 \end{aligned}
 \label{eq:messages}
\end{equation}
where $C_s$ is the stabilizer coset of a fixed Pauli representative with
syndrome $s$, and $C_sL$ is the sector with logical label $L$. Throughout, we use the terms ``logical sector,'' ``logical coset,'' and
``logical class'' interchangeably for the set $C_sL$. The quantity $T_{s,L}$ is the peak probability within a logical
sector, whereas $M_{s,L}$ is its total probability mass. Through the
Nishimori condition, the low-noise limit $t\to0$ corresponds to the
zero-temperature limit of the effective decoding model
\cite{BeverlandEtAl2019,Wichette2026}.

Writing $w_{\min}(s)$ for the minimum physical error weight compatible with
syndrome $s$, and retaining every maximizer under each decision rule, we
define $m$ as the smallest value of $w_{\min}(s)$ for which the MP and MLD
logical winner sets are disjoint for all sufficiently small $t$. Since the
winner sets are disjoint, no syndrome-dependent MP tie resolution, even the
optimal one, can eliminate the separation.

Code distance imposes a universal lower bound on this onset. For every
stabilizer code under code-capacity Pauli noise, a strict decision split is
impossible below $h\equiv\lceil d/2\rceil$, where $d$ is the code distance,
the minimum weight of a nontrivial logical Pauli. Distance alone does not
determine when the split occurs. The number of
additional error-weight shells beyond $h$ on which MP and MLD remain
equivalent as $t\to0$ defines the rigidity depth $r$, so that
\begin{equation}
 m=h+r,
 \label{eq:rigidity-intro}
\end{equation}
where the code-dependent integer $r$ is the rigidity depth defined
precisely in Sec.~\ref{sec:framework}. It counts the consecutive excitation
layers beginning at $h$ on which recursion, topology, or parity-check algebra
still prevents configurational entropy from changing the logical decision.
For every family studied here, we certify all layers
$h,\ldots,h+r-1$ as entropy neutral and construct a strict split witness at
$h+r$. The six code families therefore form a three-level physical
classification. Planar surface codes and the two recursively concatenated
families have $r=0$. Odd-distance square toric codes
and the Gross $[[144,12,12]]$ quantum low-density-parity-check (qLDPC) code
have $r=1$. One separable family $\mathcal H_q$, admitting both hypergraph
product and bivariate bicycle descriptions, has $r=2$. Code distance
determines when entropy is first allowed to matter, while the rigidity depth
determines when it actually changes the logical decision.

This finite-size hierarchy complements threshold physics. Thresholds
locate the critical noise separating ordered and disordered phases as
$d\to\infty$, whereas $m$ identifies, at fixed code size and as $t\to0$,
the first microscopic excitation layer at which MP becomes strictly
suboptimal relative to MLD. With below-threshold operation now
demonstrated for surface codes on superconducting processors
\cite{Google2025} and color codes on trapped-ion platforms
\cite{LavasaniEtAl2026}, finite-size memories are increasingly operated
in the subthreshold regime. Because identifying a most probable error is
NP-complete in general whereas evaluating a full logical-coset partition
function is \#P-hard \cite{IyerPoulin2015}, it is useful to know when
most-probable-error decoding suffices and when coset summation is
required. The leading gap $K_m t^m$ provides an exact benchmark for this
performance penalty: it determines the physical error rate below which
degeneracy corrections are negligible to a prescribed accuracy, and
Sec.~\ref{sec:operational} verifies that the resulting two-term
expansion is accurate to within ten percent at experimentally realized
noise levels.

Although the logical failure probabilities of both decoders vanish as
$t\to0$, their subthreshold performance is distinguished by their
leading power-law behavior in $t$. In the zero-temperature limit along
the Nishimori line, MLD does not collapse onto MP: each logical sector
retains a residual zero-temperature degeneracy, determined by the
multiplicity of its own minimum-weight errors. Consequently, the
failure-rate gap between the two decoders has the exact leading form
$K_m t^m$ with $K_m>0$, where the exponent $m$ and coefficient $K_m$
quantify the subthreshold advantage of degeneracy-aware decoding over
any decoder restricted to the most probable error. By contrast, at
finite $t$, the large multiplicity of a higher-weight shell can outweigh
a lower-weight peak, causing the two decoders to disagree even on a
specific error of weight $w<d/2$ that MP corrects and MLD
misidentifies. The bound $m\geq h$ and the relation $m=h+r$ are
therefore strictly asymptotic ($t\to0$) statements, and
Sec.~\ref{sec:operational} quantifies the finite-noise window over which
this leading-order behavior dominates.

We use code capacity Pauli noise as a controlled physical baseline rather
than as a proxy for a particular hardware implementation. It isolates the
intrinsic entropy generated by code geometry, logical topology, and parity
check algebra before a syndrome extraction circuit introduces scheduling
dependent fault multiplicities, propagated hook errors, and temporal
boundaries. The exact relation $m=h+r$ therefore provides a theorem level
reference point against which changes caused by a specified circuit or
decoder can be assessed.

The article is organized by this hierarchy. Sec.~\ref{sec:framework}
establishes the universal half-distance bound and the fixed-weight
syndrome-fiber framework. Sec.~\ref{sec:mechanisms} realizes successive
rigidity depths through exact concatenation recursions, boundary
conditions, and qLDPC parity-check algebra.
Sec.~\ref{sec:operational} converts the onset into a finite noise operational
benchmark. Sec.~\ref{sec:discussion} explains the conditions required to
extend the framework to detector fault hypergraphs and discusses the
implications for scalable quantum memories.

\section{Universal obstruction and the rigidity depth}
\label{sec:framework}

The baseline $h$ introduced above follows from a universal obstruction. For
every distance-$d$ stabilizer code and every interior Pauli-noise direction,
\begin{equation}
 m \ge h .
 \label{eq:lower}
\end{equation}
Indeed, for a physical Pauli error $E$, write
$\operatorname{wt}(E)$ for the number of physical qubits on which $E$
acts nontrivially and $\operatorname{syn}(E)$ for its binary stabilizer
syndrome, whose $i$th component is $0$ if $E$ commutes with stabilizer
generator $g_i$ and $1$ if it anticommutes. The minimum physical-error
weight consistent with syndrome $s$ is then
$w_{\min}(s):=\min\{\operatorname{wt}(E):\operatorname{syn}(E)=s\}$.
For a syndrome $s_\star$ realizing the onset,
$m=w_{\min}(s_\star)$. Choose $L_{\rm MP}$ and $L_{\rm MLD}$ from its
disjoint MP and MLD winner sets. The MP-winning sector contains a weight-$m$
representative and therefore contributes at order $\Theta(t^m)$. The
MLD-winning sector must also begin at weight $m$: a sector whose first
nonzero contribution is only $O(t^{m+1})$ cannot exceed a sector with a
positive $\Theta(t^m)$ contribution as $t\to0$. Calling the two weight-$m$
representatives
$E_{\rm MP}$ and $E_{\rm MLD}$, their common syndrome and distinct logical
labels make $E_{\rm MP}^{\dagger}E_{\rm MLD}$ a nontrivial logical Pauli.
We therefore have
\begin{equation*}
 d\leq \operatorname{wt}(E_{\rm MP}^{\dagger}E_{\rm MLD})
 \leq \operatorname{wt}(E_{\rm MP})+\operatorname{wt}(E_{\rm MLD})=2m.
\end{equation*}
Thus all syndrome layers below $h$ are universally protected: distance
forbids an equal-order logical competitor before
configurational multiplicity can alter the decision.  This obstruction is
independent of code geometry and decoder implementation and holds for every
interior Pauli-noise direction, including arbitrarily anisotropic ones.
The protection below $h$ is nonetheless strictly asymptotic: it rests on
the ordering of powers of $t$ as $t\to0$. Fix a syndrome $s$ whose
minimum-weight compatible error $E_1$ has weight $w_1<h$ and lies in
sector $A$; every error of a competing sector $B$ with the same syndrome
then has weight at least $w_2\geq d-w_1>w_1$. At finite noise strength,
sector $B$ may carry a multiplicity $N_B$ so large that its mass
$M_{s,B}(t)\sim N_Bt^{w_2}$ overtakes the leading mass
$M_{s,A}(t)=\Theta(t^{w_1})$ once $t$ exceeds the crossover scale
$t_\ast\sim N_B^{-1/(w_2-w_1)}$. On the specific instance $E_1$, MP then
identifies the correct sector while MLD selects sector $B$ and fails,
even though MLD still minimizes the ensemble-averaged logical failure
rate.

The half distance bound identifies the first layer at which a split is
possible, but it does not determine whether the onset occurs there or is
delayed by additional rigidity. To resolve the onset exactly, we expand the
probability mass of each logical sector in powers of $t$. For a code on $n$
qubits, this expansion terminates at physical error weight $n$, so the
comparison is finite.
For each logical sector, let $P_{s,L}^{(j)}$ denote its total
composition-weighted multiplicity at physical error weight $j$, as defined in
Supplemental Material Sec.~S1 \cite{Supplement}. Arranging these coefficients
in increasing physical weight gives the \emph{coset enumerator spectrum}
$\boldsymbol P_{s,L}=(P_{s,L}^{(w_s)},P_{s,L}^{(w_s+1)},\ldots,
P_{s,L}^{(n)})$, where $w_s=w_{\min}(s)$. After evaluation along the fixed
Pauli ray $(x,y,z)$, MLD selects, for sufficiently small $t$, the logical
sectors whose spectra are maximal in the lexicographic order.

MP instead compares the largest probability of a single microscopic error.
We write $T_{s,L}$ for the leading composition monomial of $T_{s,L}(t)$,
the largest monomial among the minimum weight representatives in sector
$L$. Writing
$\operatorname*{LexMax}_{L}\boldsymbol P_{s,L}$ for the set of
lexicographic maximizers, a strict eventual split occurs precisely when
\begin{equation}
 \arg\max_L T_{s,L}
 \cap
 \operatorname*{LexMax}_{L}\boldsymbol P_{s,L}
 =
 \varnothing .
 \label{eq:spectral-criterion}
\end{equation}
This comparison retains the higher-weight coefficients needed to resolve ties
that the minimum-weight layer alone cannot distinguish.
The set-valued winner definition removes dependence on an arbitrary tie
convention, while the coset enumerator spectrum determines the genuine
eventual MLD winner when successive low-weight coefficients are tied.

The leading layer of this comparison admits a structural representation in
terms of Pauli-type weight enumerators within syndrome cosets. At physical
error weight $w$, collect the weight-$w$ physical errors compatible with a
fixed syndrome $s$, partitioned by their logical labels and counted once;
we call this set the \emph{syndrome fiber} of $s$ at weight $w$. If two
errors $E$ and $F$ in the same fiber belong to different logical classes,
their product $E^\dagger F$ factorizes a nontrivial logical Pauli into its
endpoint errors, with $\operatorname{wt}(E^\dagger F)\leq2w$. Physically,
MP compares the most probable representative of the competing sectors,
whereas MLD also counts the low-weight representatives. A single
high-probability error can therefore make one sector the MP winner, while
a larger collection of lower-probability errors can make another sector
the MLD winner [Fig.~\ref{fig:rigidity}(a)].

This structure exposes a second invariant beyond distance. We call the
candidate layer $h+j$ \emph{entropy neutral} if, for every syndrome $s$ with
$w_{\min}(s)=h+j$, an exact certificate shows that the MP and MLD
winner sets have a nonempty intersection. Such a certificate may take the
form of a pairing that preserves composition, a fiber with one
representative in each competing class, an invariant return of the exact
concatenation recursion, or an algebraic packing
argument; in every case, neutrality is established exactly rather than
inferred from sampling. The \emph{rigidity depth} $r$ is the number of
consecutive error-weight shells, beginning at $h$, on which MP and MLD
remain equivalent as $t\to0$. If the
layers $h,\ldots,h+r-1$ are certified neutral and a strict witness with
disjoint MP and MLD winner sets exists at $h+r$, then the onset is exactly
$m=h+r$. Thus $r$ counts how many excitation layers allowed by distance
remain unable to convert configurational entropy into a logical decision
[Fig.~\ref{fig:rigidity}(b)].

This yields a three-part structure. Distance
obstructs every split below $h$; exact pairing or return maps transfer
neutrality or a strict split between related fibers; and geometric or
algebraic rigidity determines how many candidate layers remain neutral before
the first imbalance appears. Supplemental Material Sec.~S2 states the general
certificate criterion, and Secs.~S3--S8 establish the layer assignments used
here \cite{Supplement}.
The resulting neutral layer hierarchy is summarized in
Fig.~\ref{fig:rigidity}(b). Its recursive, geometric, and algebraic
realizations are developed in Sec.~\ref{sec:mechanisms}.

\begin{figure}[t]
\centering
\noindent
\begin{minipage}[t]{\columnwidth}
  \vspace{0pt}\noindent%
  \begin{minipage}[t]{1.60em}\vspace{0.12em}\makebox[\linewidth][c]{\normalsize (a)}\end{minipage}%
  \begin{minipage}[t]{\dimexpr\linewidth-1.60em\relax}\vspace{0pt}%
    \includegraphics[width=\linewidth]{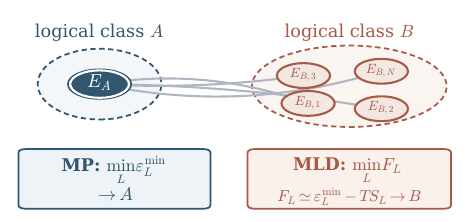}%
  \end{minipage}%
\end{minipage}

\vspace{0.25em}

\noindent
\begin{minipage}[t]{\columnwidth}
  \vspace{0pt}\noindent%
  \begin{minipage}[t]{1.60em}\vspace{0.12em}\makebox[\linewidth][c]{\normalsize (b)}\end{minipage}%
  \begin{minipage}[t]{\dimexpr\linewidth-1.60em\relax}\vspace{0pt}%
    \includegraphics[width=\linewidth]{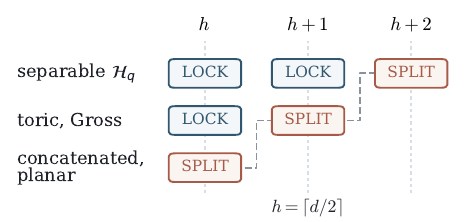}%
  \end{minipage}%
\end{minipage}
\caption{\textbf{Competition between peak probability and coset mass, and the
rigidity depth.}
(a) In the syndrome fiber at fixed $(s,w)$, MP selects the class containing the most
probable microscopic representative, whereas MLD can select a competing class
whose larger multiplicity lowers its logical sector free energy. Gray edges
join representatives that share a syndrome but belong to different logical
classes. (b) Exact certificates for neutral layers and the first strict witness give
$m=h+r$: concatenated and planar families have $r=0$, toric and Gross
memories have $r=1$, and the separable family $\mathcal H_q$ has $r=2$.
Here $\mathcal H_q$ is a single family admitting both hypergraph-product and
bivariate bicycle descriptions.}
\label{fig:rigidity}
\end{figure}

\section{Mechanisms of the rigidity depth}
\label{sec:mechanisms}

The mechanisms below form a single hierarchy of constraints on low order
multiplicity. Exact recursive closure transfers an onset witness through every
concatenation depth. Boundary geometry then distinguishes open boundaries
that permit immediate path proliferation from periodic homology that
protects the first candidate layer. Sparse parity-check algebra imposes the strongest
restriction by matching the Pauli compositions of entire low-weight fibers.
In every case, an exact neutrality certificate excludes separation between MP
and MLD throughout the protected layers, and a constructive witness identifies
the first imbalanced layer. Figs.~\ref{fig:rg-invariants},
\ref{fig:boundary-mechanism}, and \ref{fig:algebraic-mechanisms} follow this
progression from recursive transfer to geometric and algebraic rigidity.

\begin{figure}[t]
\centering
\includegraphics[width=\columnwidth,trim=0 7mm 0 0,clip]{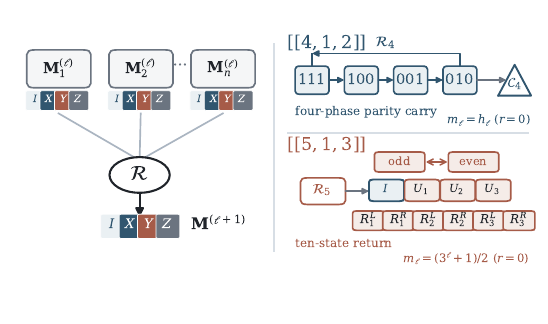}
\caption{\textbf{Exact recursive returns at rigidity depth $r=0$.}
For the $[[4,1,2]]$ family, the parity vector of the three nontrivial
logical sector weights follows a four phase cycle that closes under the exact
RG map. For the $[[5,1,3]]$ family, the identity state, three overlap states
$U_a$, and their six directed resolutions $R_a^L,R_a^R$ form a ten state
alphabet that also closes under the RG map. In both families the first
distance allowed layer remains a strict witness after every recursive step,
so $m_\ell=h_\ell$ at all depths.}
\label{fig:rg-invariants}
\end{figure}

\subsection{Exact concatenation recursions and their invariants}
Recursive concatenation induces an exact recursion on the four logical
sector partition functions
\cite{RahnDohertyMabuchi2002,Poulin2006,EvansStephens2012}. In
statistical-mechanics language this recursion is a block-spin
renormalization group (RG) map on the hierarchical lattice defined by
concatenation: optimal decoding of concatenated codes can be reformulated
as an RG flow on such a lattice \cite{EnglishEtAl2026RG}, and the
recursions of
Refs.~\cite{RahnDohertyMabuchi2002,Poulin2006,EvansStephens2012} are its
finite-dimensional representation through logical sector partition
functions. The MLD recursion sums the products of
all compatible child messages, whereas its MP counterpart retains the
largest compatible product. At low noise, projection onto the
\emph{leading term} records the first nonzero exponent and its positive
coefficient. The exponents obey the rules of tropical (min,+) algebra
\cite{LiuWangZhang2021}: addition keeps the smaller
exponent and multiplication adds exponents. When exponents tie, the exact
recursion also adds their coefficients and therefore retains the
degeneracy discarded by ordinary tropicalization. Tensor enumerator
contraction supplies this
coefficient-sensitive recursion
\cite{FerrisPoulin2014,CaoLackey2023,CaoEtAl2024}.

The problem at arbitrary depth reduces to a finite certificate based on an
invariant cone: a finite set of branch selection inequalities is mapped into
itself by exact seed
contraction, and a strict return propagates an entering witness by induction.
The finite state labels retain only the parity or overlap data needed to
determine the leading coefficient. Closure means that no new leading pattern
appears after another concatenation step. Once a half-distance split witness
returns to the same certified state, induction propagates that witness to
every depth. The first distance-allowed layer is therefore already
imbalanced, giving $r=0$.
The general finite-certificate transfer principle is established in
Supplemental Material Sec.~S2, while the exact recurrences, invariant regions,
and certificates at every depth for the $[[5,1,3]]$ and $[[4,1,2]]$ families are
given in Secs.~S3 and S4, respectively \cite{Supplement}.

Two inequivalent recursive mechanisms realize rigidity depth $r=0$ and saturate
Eq.~(\ref{eq:lower}) [Fig.~\ref{fig:rg-invariants}]. The first uses the
standard $[[4,1,2]]$ subcode obtained by fixing one logical degree of freedom
of $[[4,2,2]]$, expressed in a local Clifford frame adapted to the noise. If
$(0,a_\ell,b_\ell,c_\ell)$ are the minimum weights of its four zero-syndrome
logical sectors, the recursion induces the parity cycle
$\pi_\ell=(a_\ell,b_\ell,c_\ell)\bmod2:
111\to100\to001\to010\to111$. This parity determines whether the two
half-logical representatives in the witness are disjoint or overlap on one
site. The four phases and three nontrivial target sectors therefore reduce the
all-depth support construction to twelve finite recipes. Their exact return
proves onset at half distance for every depth inside a certified anisotropic
cell, demonstrating that the invariant cone mechanism persists away from
symmetric noise.

For the cyclic perfect $[[5,1,3]]$ code, let
$a\in\{1,2,3\}$ label the target nontrivial logical Pauli. The state $U_a$
records two same-syndrome half-distance representatives whose product is a
minimum logical of type $a$ and whose supports overlap on one site. The states
$R_a^L$ and $R_a^R$ are the two directed resolutions of that overlap, with
endpoint weights $(h_\ell,h_\ell-1)$ and
$(h_\ell-1,h_\ell)$, respectively. Together with the identity state, these
give a ten-state alphabet. An exact five-child substitution maps this alphabet
into itself, while two phases track the alternating logical labels of the
endpoints. The substitution preserves the common syndrome, target logical
product, half-distance endpoint weights, and single overlap, so the same
witness type returns at the next depth. This closure proves the exact onset at
half distance for every depth $\ell\geq2$ throughout the open analytic region shown in
Fig.~\ref{fig:onsets}, which includes points with a normalized witness
margin of order one. Thus both recursive mechanisms certify half-distance
saturation at every concatenation depth throughout explicit open noise
regions. This complements earlier studies of recursive decoding and phase
diagrams by determining the first tie-safe decision split between MP and MLD
\cite{SommersHuseGullans2025}.

\subsection{Boundary conditions and low-weight degeneracy}
Open boundaries and periodic closure impose opposing constraints on logical
coset degeneracy at low excitation order
[Fig.~\ref{fig:boundary-mechanism}]. For the displayed planar witness, class
$A$ has two minimum weight representatives, while class $B$ has four
boundary-reflected completions with the same weight and syndrome. More
generally, the family of boundary-reflected paths increases the leading mass
of class $B$ without increasing
its peak probability. The first allowed layer $h$ is therefore already
imbalanced, so $r=0$.
Throughout the certified open region in Fig.~\ref{fig:onsets}, the peak favors
one class but the summed mass favors the other, attaining the universal bound.

\begin{figure}[t]
\centering
\noindent
\begin{minipage}[t]{\columnwidth}
  \vspace{0pt}\noindent%
  \begin{minipage}[t]{1.55em}\vspace{0.10em}\makebox[\linewidth][c]{\normalsize (a)}\end{minipage}%
  \begin{minipage}[t]{\dimexpr\linewidth-1.55em\relax}\vspace{0pt}%
    \includegraphics[width=\linewidth,trim=0 2mm 0 2.8mm,clip]{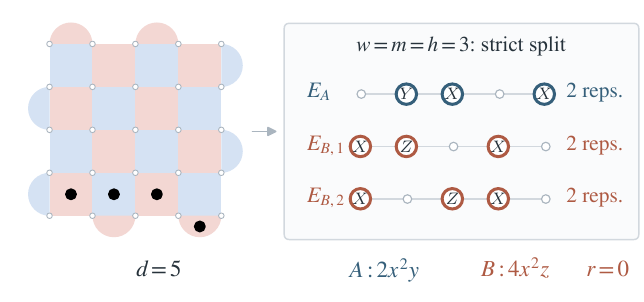}%
  \end{minipage}%
\end{minipage}
\vspace{0em}

\noindent
\begin{minipage}[t]{\columnwidth}
  \vspace{0pt}\noindent%
  \begin{minipage}[t]{1.55em}\vspace{0.10em}\makebox[\linewidth][c]{\normalsize (b)}\end{minipage}%
  \begin{minipage}[t]{\dimexpr\linewidth-1.55em\relax}\vspace{0pt}%
    \includegraphics[width=\linewidth,trim=0 2mm 0 2.8mm,clip]{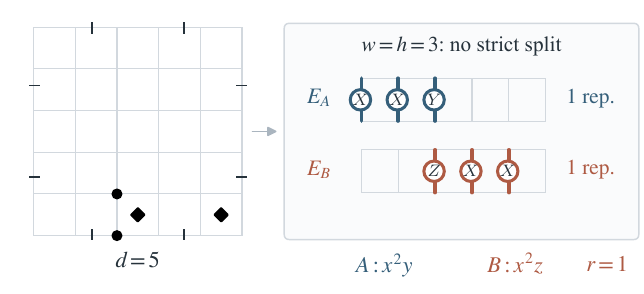}%
  \end{minipage}%
\end{minipage}
\vspace{0em}

\noindent
\begin{minipage}[t]{\columnwidth}
  \vspace{0pt}\noindent%
  \begin{minipage}[t]{1.55em}\vspace{0.10em}\makebox[\linewidth][c]{\normalsize (c)}\end{minipage}%
  \begin{minipage}[t]{\dimexpr\linewidth-1.55em\relax}\vspace{0pt}%
    \includegraphics[width=\linewidth,trim=0 2mm 0 2.8mm,clip]{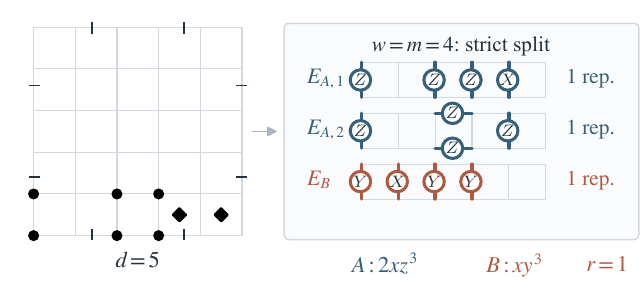}%
  \end{minipage}%
\end{minipage}
\vspace{-0.20em}
\caption{\textbf{Boundary path multiplicity and periodic loop rigidity at $d=5$.}
The labels ``1 rep.'' and ``2 reps.'' count distinct physical errors represented
by each displayed composition prototype. Errors counted by the same label have
the same syndrome, logical class, Pauli composition, and probability, but are
not identical physical Pauli operators.
(a) For the displayed planar syndrome, open boundaries permit four
minimum weight representatives in class $B$ but only two in class $A$.
The resulting leading masses are $P_A=2x^2y$ and $P_B=4x^2z$, so the first
allowed layer $h=3$ can already split the peak and mass decisions, giving
$r=0$. (b) For the displayed toric syndrome at $h=3$, each competing
logical class has one minimum weight representative. The leading mass
therefore equals the peak classwise, and no strict split occurs. Black circles
and diamonds mark the violated vertex and plaquette checks that form the
common syndrome. (c) At $m=4$, a local deformation produces two class $A$
representatives and one class $B$ representative, with leading masses
$P_A=2xz^3$ and $P_B=xy^3$. This is the first strict toric split and gives
$r=1$. The line segments in the right boxes show the lattice support of
physical Pauli errors.}
\label{fig:boundary-mechanism}
\end{figure}

Along the low noise Pauli ray
$(p_X,p_Y,p_Z)=t(x,y,z)$, the peak and mass select different logical classes
whenever $y>x$, $y>z$, and
$2z(1-x^2/y^2)>y$. An enumeration based on the reflection principle pairs
paths at their first boundary crossing and establishes this result for every
odd distance.

Periodic closure restores neutrality at $h$. Removing boundary endpoints forces the two
competitors at $h$ into complementary segments of a noncontractible cycle.
For each fixed syndrome, each logical class then contains exactly one geodesic
representative, giving the matched fiber in
Fig.~\ref{fig:boundary-mechanism}(b). Its leading mass equals its peak within
each class, so the layer is entropy neutral. A transverse excursion changes
the combined cycle by two edges. Each competitor must therefore gain one
physical error before this deformation can appear. At $h+1$, the deformation
produces two class $A$ representatives and one class $B$ representative, as
shown in Fig.~\ref{fig:boundary-mechanism}(c). This is the first possible
imbalance, hence $r=1$, and it reverses the decisions throughout the toric
region in Fig.~\ref{fig:onsets}. The contrast between boundary conditions is
therefore
\begin{equation}
 m_d^{\mathrm{planar}}=\frac{d+1}{2},
 \qquad
 m_d^{\mathrm{toric}}=\frac{d+3}{2},
 \qquad d=3,5,\ldots .
 \label{eq:topological-onsets}
\end{equation}
The $d=3$ case is verified directly. The planar boundary-path enumeration and the
proofs of toric loop rigidity and deformation are given in Secs.~S5 and S6 of the
Supplemental Material \cite{Supplement}, respectively. These results are
restricted to the specified code families and code capacity Pauli noise; the
arguments for all distances are analytic, with finite enumeration used only for the
$d=3$ base case and for consistency checks.

For $d=2a+1\geq5$, the local toric deformation has two representatives of
minimum weight in class $A$ and one in class $B$. Consequently,
$T_A=x^{a-1}z^3$ and $P_A=2x^{a-1}z^3$, whereas
$T_B=P_B=x^{a-1}y^3$. Thus $z<y$ makes class $B$ the peak winner, while
$y^3<2z^3$ makes class $A$ the mass winner. Both monomials have total degree
$(a-1)+3=a+2=h+1$. Their common factor $x^{a-1}$ cancels from the
comparison between peak and mass, showing that increasing $d$ lengthens only the shared
geodesic segment without changing the local splitting condition.

\subsection{Algebraic rigidity in qLDPC memories}
In topological memories, geometry controls how many error paths of low weight can
realize one syndrome. In qLDPC memories, sparse parity checks play the analogous
role: they can exclude logical factorizations near the minimum or force competing
representatives into fibers with matched compositions. Their multiplicity then
cannot change the leading ordering of coset masses. The number of consecutive
layers protected in this way is the algebraic contribution to the
rigidity depth $r$ [Fig.~\ref{fig:algebraic-mechanisms}].

The two qLDPC results below have different scopes. The Gross
$[[144,12,12]]$ memory $\mathcal G$ is a prominent finite BB code relevant to
current low-overhead hardware development
\cite{BravyiEtAl2024Gross,YoderEtAl2025Tour}. For this specific instance,
$m(\mathcal G)=7$ is an exact theorem. Its half-distance layer $h=6$ is
protected by an algebraic property of the code: every minimum-weight
($d=12$) logical Pauli is pure in Pauli type, meaning that it uses
only one nonidentity Pauli species. Because the product of any two
competing weight-6 representatives sharing a syndrome is a weight-12
logical operator, the two complementary halves must have identical Pauli
composition, rendering every half-distance syndrome fiber composition
matched and entropy neutral.

Establishing this onset does not require solving full degenerate MLD,
which is \#P-hard in general \cite{IyerPoulin2015} and would entail
comparing all $4^{12}$ stabilizer-coset sums for $\mathcal G$. Because
the onset is decided as $t\to0$, only the lowest-weight error shells
enter, reducing the proof to bounded exact computations on 144 qubits
(Supplemental Material Sec.~S8 \cite{Supplement}). To prove the
Pauli-purity property at $h=6$, we partition the search for a mixed-type
weight-12 logical operator into 24 Boolean satisfiability (SAT) sectors,
all of which are exhaustively unsatisfiable; an independent enumeration
lists all 3780 minimum-weight logical representatives and audits every
weight-6 syndrome generated by their half-weight splits. At the next
shell, a second SAT check excludes nontrivial centralizers of weight 13,
confining the weight-7 competition to factorizations of weight 12 or 14.
For our witness syndrome, we pose the fixed-syndrome problem over the 132
parity checks as an integer linear program (equivalently, a SAT instance)
in the Pauli occupation variables, just as in minimum-weight decoding.
This instance is unsatisfiable up to weight 6 and yields exactly three
operators at weight 7: one representative that wins the peak comparison
and a doublet in a competing sector that wins the mass comparison,
producing the first imbalanced fiber \cite{ZhaoYanZenodo2026}. By
contrast, the separable family $\mathcal H_q$ requires no such
computation: $m(\mathcal H_q)=q+2$ is proved analytically for all $q$,
with computation entering only the $q=1$ base case.

For the same Pauli noise composition, the onset fiber at weight 7 contains
one representative in class $A$ and two in class $B$. Consequently,
$T_A=P_A=x^5y^2$, whereas $T_B=x^5z^2$ and $P_B=2x^5z^2$. Thus $y>z$
makes class $A$ the peak winner, while $y^2<2z^2$ makes class $B$ the mass
winner. Both monomials have total degree $7=d/2+1$; their common factor $x^5$
cancels from the decision comparison, isolating a local composition defect on
top of the rigid global structure of minimum-weight logical operators
[Fig.~\ref{fig:algebraic-mechanisms}(b)].

The separable family $\mathcal H_q=[[18q^2,8,2q]]$, $q\geq1$, which admits
both hypergraph product and bivariate bicycle descriptions
\cite{TillichZemor2014}, instead gives a theorem for every $q$:
$m(\mathcal H_q)=q+2$. Here $h=q$. A Pauli systolic gap of four excludes
logical factorizations of low weight that could create an unmatched fiber.
Together with a packing of disjoint parity detectors, it forces the remaining
competitors at weights $q$ and $q+1$ into pairs with identical Pauli
composition. These two layers are therefore entropy neutral. At weight
$q+2$, a periodic-chain configuration
saturates the detector packing and produces the first imbalanced fiber. It
contains two representatives $E_{A,1}$ and $E_{A,2}$ in logical class $A$,
and singleton representatives $E_B$ and $E_C$ in two distinct competing
classes. Hence $r=2$. The construction for every $q$, the systolic gap lemma,
the detector packing proof, and the
certified base case are given in Supplemental Material Sec.~S7
\cite{Supplement}; Fig.~\ref{fig:algebraic-mechanisms}(a) illustrates the packing and
the onset representatives. The two results from algebraic rigidity are
\begin{equation}
 m(\mathcal G)=7=\frac{d}{2}+1,
 \qquad
 m(\mathcal H_q)=q+2=\frac{d}{2}+2 .
 \label{eq:qldpc-onsets}
\end{equation}
For $q\geq2$, the two class $A$ representatives give
$T_A=x^3z^{q-1}$ and $P_A=2x^3z^{q-1}$, whereas the two singleton classes
give $T_B=P_B=T_C=P_C=xy^2z^{q-1}$. Thus $y>x$ makes classes $B$ and $C$
the peak winners, while $y^2<2x^2$ makes class $A$ the mass winner;
equivalently,
$x<y<\sqrt2x$. Both monomials have total degree
$1+2+(q-1)=q+2=d/2+2$, while their common factor $z^{q-1}$ cancels from the
comparison. The shared factor is the periodic chain; the local endpoint cluster,
rather than the bulk length, determines the strict decision
[Fig.~\ref{fig:algebraic-mechanisms}(a)]. For $q=1$, the same onset
$m=3$ holds because the mass-winning class has multiplicity six rather
than two.

\begin{figure}[t]
\centering
\noindent
\begin{minipage}[t]{\columnwidth}
  \vspace{0pt}\noindent%
  \begin{minipage}[t]{1.55em}\vspace{0.15em}\makebox[\linewidth][c]{\normalsize (a)}\end{minipage}%
  \begin{minipage}[t]{\dimexpr\linewidth-1.55em\relax}\vspace{0pt}%
    \includegraphics[width=\linewidth,trim=0 10mm 0 0,clip]{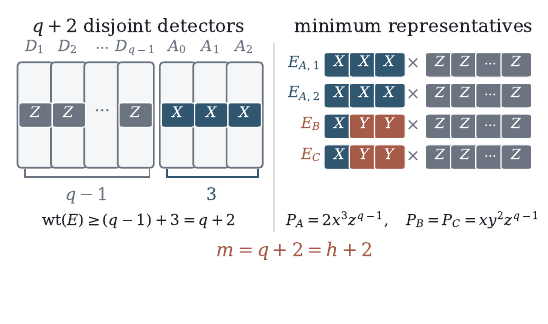}%
  \end{minipage}%
\end{minipage}
\vspace{0.05em}

\noindent
\begin{minipage}[t]{\columnwidth}
  \vspace{0pt}\noindent%
  \begin{minipage}[t]{1.55em}\vspace{0.15em}\makebox[\linewidth][c]{\normalsize (b)}\end{minipage}%
  \begin{minipage}[t]{\dimexpr\linewidth-1.55em\relax}\vspace{0pt}%
    \includegraphics[width=\linewidth,trim=0 10mm 0 0,clip]{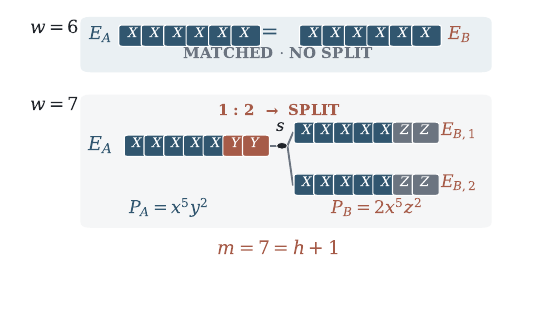}%
  \end{minipage}%
\end{minipage}
\vspace{-0.25em}
\caption{\textbf{Algebraic rigidity in qLDPC memories.}
(a) In the separable family, detector packing forces every competing error
configuration with unequal Pauli composition to have weight at least $q+2$. At that layer, two
representatives in class $A$ compete with singleton representatives in two
other classes. Their leading masses are $P_A=2x^3z^{q-1}$ and
$P_B=P_C=xy^2z^{q-1}$, yielding the first strict split at
$m=q+2=h+2$. (b) In the Gross $[[144,12,12]]$ code, Pauli purity keeps the
$w=6$ fibers composition matched. At $w=7$, one representative wins the peak
comparison while a doublet in the competing class wins the mass comparison.
The first strict split therefore occurs at $m=7=h+1$.}
\label{fig:algebraic-mechanisms}
\end{figure}

\begin{table*}[t]
\caption{Classification of the exact onset $m=h+r$ by rigidity depth. The planar and toric rows assume odd $d\geq3$; the HGP/BB row holds for $q\geq1$.}
\label{tab:summary}
\centering
\setlength{\tabcolsep}{7.5pt}
\renewcommand{\arraystretch}{1.08}
\begin{tabular}{lllll}
\toprule
Family & Distance & $r$ & Exact onset $m=h+r$ & Neutral layer certificate \\
\midrule
$[[4,1,2]]^{\circ\ell}$ & $2^\ell$ & $0$ & $h$ & four-phase parity cycle \\
$[[5,1,3]]^{\circ\ell}$ & $3^\ell$ & $0$ & $h$ & odd overlap and invariant return \\
Rotated planar surface & odd $d$ & $0$ & $(d+1)/2$ & boundary path multiplicity (reflection principle) \\
Square toric & odd $d$ & $1$ & $(d+3)/2$ & noncontractible loop rigidity \\
Gross $[[144,12,12]]$ & $12$ & $1$ & $7$ & Pauli-type purity of minimum logicals \\
HGP/BB $\mathcal H_q$ & $2q$ & $2$ & $q+2$ & systolic gap and detector packing \\
\bottomrule
\end{tabular}
\end{table*}

Table~\ref{tab:summary} and Fig.~\ref{fig:rigidity}(b) summarize the
resulting hierarchy and contrast the mechanisms that control the
rigidity depth $r$ across all six families. In topological codes, the
onset of degeneracy is governed by global boundary conditions rather
than bulk geometry: open boundaries permit reflected paths that split MP
and MLD immediately at $h$ ($r=0$), whereas periodic closure locks
competing errors into complementary segments of a noncontractible cycle
and delays the split to $h+1$ ($r=1$). In the $[[144,12,12]]$ Gross
code, 3780 minimum-weight logical operators generate 1{,}452{,}000
multiclass weight-6 syndromes, yet the entire half-distance shell
remains entropy neutral ($m=7$, $r=1$) because every weight-12 logical
is pure in Pauli type, forcing its two weight-6 halves to share an
identical Pauli composition. Finally, in the separable HGP/BB family
$\mathcal H_q$, a Pauli systolic gap of four combines with disjoint
detector packing to protect two full shells above $h$ ($r=2$) for all
$q\ge1$.

\begin{figure}[t]
\centering
\includegraphics[width=\columnwidth,trim=0 3mm 0 0,clip]{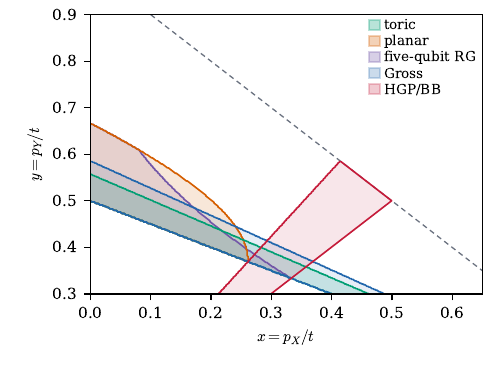}
\par\vspace{-0.60em}
\caption{\textbf{Certified noise regions.} Sufficient open regions in the
interior Pauli simplex for the five-qubit concatenation, planar, toric, separable HGP/BB,
and Gross-code results. The gray dashed line is the simplex boundary $z=0$.
The colored boundaries are sufficient certificates, not complete phase
boundaries.}
\label{fig:onsets}
\end{figure}

\section{Operational benchmark and finite-noise calibration}
\label{sec:operational}
The composition spectrum converts the onset into an exact operational
benchmark for the contribution of logical-sector degeneracy to the decision
gap between best-tie MP and MLD decoding. Using the notation of
Sec.~\ref{sec:framework}, write $w_s=w_{\min}(s)$,
${\cal A}_s=\arg\max_L T_{s,L}$ for the MP-optimal logical set, and
$P_{s,L}\equiv P_{s,L}^{(w_s)}$ for the leading coset-mass coefficient.
Best-tie MP chooses, independently for each syndrome, the member of
${\cal A}_s$ with the largest $P_{s,L}$. Its remaining leading disadvantage
is therefore
$\kappa_s=\max_L P_{s,L}-\max_{L\in{\cal A}_s}P_{s,L}\geq0$.
This is the smallest leading disadvantage attainable within the MP class:
any practical tie-breaking heuristic only enlarges it.
For a decoder $D$, denote its total logical failure probability by
$P_{\mathrm{fail}}^D$. For the six realizations studied here, the first strict
split occurs at $m=h+r$, and therefore
\begin{align}
 P_{\mathrm{fail}}^{\mathrm{MP}}(\text{best tie})
 -
 P_{\mathrm{fail}}^{\mathrm{MLD}}
 &=
 K_mt^m+O(t^{m+1}),
 \nonumber\\
 K_m
 &=
 \sum_{s:w_s=m}\kappa_s>0 .
 \label{eq:gap}
\end{align}
Because every syndrome contribution is nonnegative, the leading coefficient
cannot cancel, and $K_m>0$ remains irreducible even under optimal MP tie
resolution. This benchmark applies
directly to decoders whose decisions remain within the exact MP-optimal
logical set. A larger rigidity depth certifies more consecutive low-excitation
layers on which intrinsic coset entropy cannot penalize a most-probable-error
decision, but it does not by itself determine decoder runtime or circuit
performance.

Minimum-weight perfect matching (MWPM) belongs to this class only when its
matching graph and log-likelihood edge weights reproduce the relevant MP
decision; generic MWPM may incur an additional approximation gap
\cite{DennisEtAl2002,Fowler2013}. Belief propagation with ordered statistics
decoding (BP-OSD) is not generically MP either, because its approximate
marginals and ordered statistics postprocessing need not restrict its output
to the exact MP optimal set \cite{RoffeEtAl2020}. For an arbitrary decoder
$D$, define
$P_{\rm fail}^{\rm MP,*}\equiv
P_{\rm fail}^{\rm MP}(\text{best tie})$. The identity
$P_{\rm fail}^{D}-P_{\rm fail}^{\rm MLD}
=(P_{\rm fail}^{\rm MP,*}-P_{\rm fail}^{\rm MLD})
+(P_{\rm fail}^{D}-P_{\rm fail}^{\rm MP,*})$
separates the certified degeneracy contribution from an algorithm-dependent
remainder. The latter need not be nonnegative when $D$ can leave the MP set,
so $K_m$ is not asserted to equal the complete penalty of every MWPM or
BP-OSD implementation.

The certified onset witnesses provide deterministic, sampling-free tests of
degeneracy awareness. At each witness syndrome with $m=h+r$, a decoder that
captures the leading logical-sector multiplicity must reproduce the certified
MLD winner rather than the MP winner. Failure on this test identifies a
missing degeneracy contribution rather than merely a failure to sample a rare
logical event. Passing the test verifies the specified witness, not global
equivalence to MLD. The same criterion can test MWPM and BP-OSD
implementations, tensor network and learned decoders, and degeneracy-aware
postprocessing
\cite{BravyiSucharaVargo2014,TorlaiMelko2017,MeinerzEtAl2022}.

In terms of the Pauli-ray composition $(x,y,z)$, define the normalized witness
margin $\eta_s=\kappa_s/\max_L P_{s,L}$. It measures the fraction of the
mass-winning leading coefficient that remains inaccessible when the decision
is restricted to the MP-optimal set. This quantity is dimensionless and scale
independent for the five-qubit, planar, HGP/BB ($q\geq2$), and Gross
mechanisms and for the toric witness at $d\geq5$:
\begin{align}
 \eta_5
 &=
 1-\frac{x^2+y^2}{2yz},
 &
 \eta_{\mathrm{toric}}
 &=
 1-\frac{y^3}{2z^3},
 \nonumber\\
 \eta_{\mathrm{planar}}
 &>
 1-\frac{y}{2z(1-x^2/y^2)},
 &
 \eta_{\mathrm{HGP}}
 &=1-\frac{y^2}{2x^2},\nonumber\\
 \eta_{\mathrm{Gross}}
 &=1-\frac{y^2}{2z^2}.
 \label{eq:margins}
\end{align}
For the $q=1$ HGP base case, the corresponding margin is
$1-y^2/(6x^2)$ because the mass-winning class has six representatives.
Hence the advantage remains uniformly bounded away from zero on compact
subsets of every certified open region.

For the finite $t$ calibration below, we choose the shared benchmark
composition $(x,y,z)=(0.10,0.46,0.44)$, which lies in the common interior of
the certified five-qubit, planar, toric, and Gross regions in
Fig.~\ref{fig:onsets}. At this point the five-qubit, planar, and toric margins
exceed $0.4$, and $\eta_{\mathrm{Gross}}>0.45$. The HGP region does not
contain this benchmark point; its margin is instead uniformly positive on
compact subsets of its own certified cone. Supplemental Material Sec.~S9
derives the operational coefficient, while Sec.~S10 gives the rigorous tail
bounds, independent finite $t$ coset sums, and distance-resolved ten percent
validity windows \cite{Supplement}.

The theorem fixes an asymptotic exponent, while hardware and decoder
benchmarks operate at finite $t$. At physical error rates
$t\lesssim10^{-3}$, already realized in recent experiments
\cite{Google2025,LavasaniEtAl2026}, the two-term expansion reproduces the exact failure-rate
gap to within ten percent for every computed witness, and the bare leading
order $K_mt^m$ suffices for all but the largest-distance toric witness.
For the selected witness at distance $d$,
let $\delta_{s,d}(t)$ denote the exact fixed syndrome gap between best tie MP
and MLD, let $m_d$ be its onset weight, let $\kappa_d$ be its leading
coefficient, and let $n_d$ be the number of physical qubits. We evaluate
$\delta_{s,d}(t)$ by complete code-capacity coset sums, independently of the
leading formulas, and define the leading order approximation
$\mathrm{LO}_d(t)=\kappa_d t^{m_d}$.

Fig.~\ref{fig:finite-window}(a) compares the exact gap with LO and reproduces
the predicted planar $t^2$ and toric $t^3$ onsets. Panel (b) plots
$R_d(t)=\delta_{s,d}(t)/\mathrm{LO}_d(t)$ as distance increases. Its crossing
of $0.9$ defines a conservative ten percent validity window, not a decoding
threshold. Higher weight completions proliferate with distance, so the LO
validity window contracts clearly as $d$ increases.

Every onset representative also contains $n_d-m_d$ identity locations.
Retaining their exact probability gives the refined approximation
$\mathrm{LO{+}I}_d(t)=\kappa_d t^{m_d}(1-t)^{n_d-m_d}$. This expression
keeps the exact identity probability of the untouched locations and removes
the corresponding universal normalization correction from LO. It substantially
enlarges the validity windows shown in panel (c) without including the genuine
weight $m_d+1$ error layer or changing the onset $m_d$. After removal of the
dominant volume dependence, the LO+I validity window is controlled by the
discrete local spectrum at higher weight and therefore need not vary
monotonically with $d$. The exact gap has the form
$\kappa_d t^{m_d}(1-t)^{n_d-m_d}[1+c_dt+O(t^2)]$, where $c_d$ depends on the
multiplicity and Pauli composition of the next spectral layer. The smaller
residual corrections for the toric $d=7$ and $d=9$ witnesses produce the two
nearby elevated LO+I windows in panel (c), with the $d=9$ value slightly lower
because of subsequent spectral terms. This nonmonotonic variation arises from the discrete higher order spectrum at finite code size and does not imply a higher decoding threshold.

\begin{figure*}[t]
\centering
\begin{minipage}[t]{0.32\textwidth}
\vspace{0pt}\raggedright\small (a)\par\centering
\includegraphics[width=\linewidth,trim=0 1.5mm 0 1mm,clip]{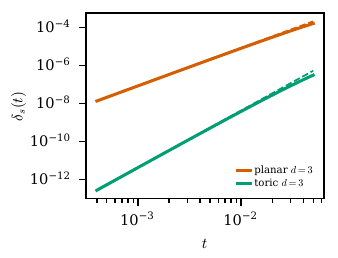}
\end{minipage}\hfill
\begin{minipage}[t]{0.32\textwidth}
\vspace{0pt}\raggedright\small (b)\par\centering
\includegraphics[width=\linewidth,trim=0 1.5mm 0 1mm,clip]{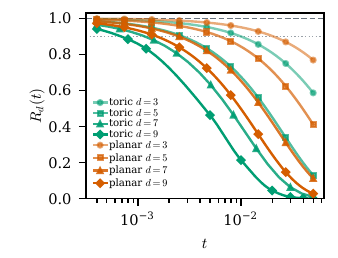}
\end{minipage}\hfill
\begin{minipage}[t]{0.32\textwidth}
\vspace{0pt}\raggedright\small (c)\par\centering
\includegraphics[width=\linewidth,trim=0 1.5mm 0 1mm,clip]{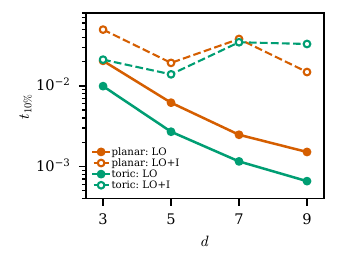}
\end{minipage}
\vspace{-0.35em}
\caption{\textbf{Finite noise calibration of the low noise onset.}
(a) Independent exact sums over complete fixed syndrome stabilizer orbits for matched $d=3$ planar and toric
witnesses (solid) and LO (dashed).
(b) $R_d(t)=\delta_{s,d}(t)/\mathrm{LO}_d(t)$ for planar
$d=3,5,7,9$ and toric $d=3,5,7,9$. The dashed line marks exact agreement,
$R_d=1$, and the dotted line marks the ten percent criterion, $R_d=0.9$.
(c) Corresponding ten percent validity scales for LO
(filled markers) and LO+I (open markers), where
$\mathrm{LO}=\kappa_d t^{m_d}$ and
$\mathrm{LO{+}I}=\kappa_d t^{m_d}(1-t)^{n_d-m_d}$.
The LO window contracts with distance, whereas the LO+I window is controlled
by the discrete local spectrum at higher weight and need not be monotonic.
Lines guide the eye and are not asymptotic fits.}
\label{fig:finite-window}
\end{figure*}

\section{Discussion and outlook}
\label{sec:discussion}

A central takeaway is that code distance alone does not dictate when
configurational entropy begins to matter in low-noise quantum
decoding. Naively, one might expect degeneracy to alter optimal
decisions as soon as competing logical sectors reach the same syndrome
at the half-distance $h=\lceil d/2\rceil$. This holds for recursively
concatenated and open-boundary codes, where $r=0$. But periodic
homology in toric codes and algebraic constraints in qLDPC memories
actively delay the first degeneracy split to one or two shells beyond
$h$. The contrast between the Gross code ($r=1$, protected by the
Pauli-type purity of its minimum logicals) and the separable family
$\mathcal H_q$ ($r=2$, protected by systolic gap and detector packing)
shows that low-weight degeneracy in qLDPC memories is controlled by
fine-grained centralizer and parity-check structure rather than by
graph sparsity alone.

The exponent $m$ describes this microscopic onset for a fixed architecture
in the low-noise limit. It complements rather than replaces the thermodynamic
threshold, which concerns the survival of an ordered phase at infinite
distance. Code-capacity Pauli noise isolates the intrinsic rigidity of the
code before syndrome extraction circuits introduce implementation dependent
fault correlations. This separation is the reason for using the model: if
circuit scheduling is introduced at the outset, intrinsic multiplicity cannot
be cleanly distinguished from multiplicity created by hook propagation,
ancilla faults, or a finite number of measurement rounds. The resulting $m$
and $r$ are therefore hardware independent reference values against which
circuit induced changes and decoder approximation can be measured. They are
not asserted to be upper or lower bounds on every circuit implementation.

For a specified syndrome extraction circuit, the fiber enumeration can be reformulated on
its $(2+1)$ dimensional detector fault hypergraph. Vertices are fault sets that
share detector outcomes but carry different logical action labels, while
elementary faults become spatial or temporal edges and hyperedges. Circuit
scheduling and hook propagation thereby reconstruct the connectivity on which
the factorization problem is posed. If fault weights are additive, the same
distance argument conditionally gives
$m_{\rm circ}\geq\lceil d_{\rm det}/2\rceil$, where the detector distance
$d_{\rm det}$ is the minimum total fault weight of an undetected nontrivial
logical action in that circuit and need not equal the code distance. If the
first $r_{\rm circ}$ candidate layers beginning at
$\lceil d_{\rm det}/2\rceil$ are separately certified neutral and a strict
witness is constructed at the next layer, then one may write
$m_{\rm circ}=\lceil d_{\rm det}/2\rceil+r_{\rm circ}$. This is a conditional
structure for a given detector hypergraph, not a circuit-level theorem proved
here. Temporal boundaries, measurement rounds, propagated hooks, and
correlated hyperedge probabilities can change both $d_{\rm det}$ and
$r_{\rm circ}$ by creating new low-weight multiplicities. Their dependence on
circuit scheduling, hook orientation, and finite measurement depth is the
subject of a separate study.

The separable HGP/BB theorem shows that LDPC sparsity and periodicity alone do
not fix the onset; its distance-normalized encoding efficiency is
$kd^2/n=16/9$, where $[[n,k,d]]$ denotes physical qubits, logical qubits, and
distance. The Gross result supplies an exact code-capacity baseline for a
prominent hardware-motivated finite BB instance with $kd^2/n=12$, relevant to
current low-overhead and modular-architecture studies
\cite{BravyiEtAl2024Gross,YoderEtAl2025Tour}.

For general constant-rate qLDPC
families \cite{PanteleevKalachev2022,LeverrierZemor2022}, we conjecture that
$r$ is governed jointly by expansion of the Tanner graph, the bipartite
qubit--check incidence graph, the stabilizer quotient, Pauli-composition
spectra, and the growth of low-excess centralizer representatives, namely
Pauli operators that commute with every stabilizer, considered modulo
stabilizers. Expansion and systolic gaps may suppress such factorizations,
whereas locally clustered low-weight centralizer classes may proliferate
them. An all-size onset theorem for a constant-rate family remains open;
the fixed-weight fiber enumeration provides a concrete, non-sampling diagnostic for
testing this design principle. Together, $r$ and $K_m$ quantify when
configurational entropy first changes the logical decision and how strongly
it affects finite-noise performance, providing an exact baseline for
determining when degeneracy-aware decoding is required. That requirement
carries a computational price: full degenerate MLD sits at the expensive
end of the decoding cost hierarchy, and skipping the sector summation is
attractive precisely when it is safe. Our closed-form results delineate this boundary: below the onset weight
the two rules share every logical winner, whereas above it the best
attainable most-probable-error rule fails more often than MLD by
$K_m t^m$ in leading order, already quantitatively reliable at physical
error rates $t\lesssim 10^{-3}$ (Sec.~\ref{sec:operational}). As
hardware continues to
lower those rates, the onset exponent separates the range in which the
inexpensive rule is provably sufficient from the range in which
degeneracy awareness measurably pays.

\begin{acknowledgments}
This work was supported by the National Natural Science Foundation of China (Grants No. 92476206 and No. 12322413) and the Beijing Natural Science Foundation (Grants No. JQ25014). Y.~Z. thanks his wife, Ling Wang, for her patience and support throughout this work.

\end{acknowledgments}

\section*{Data Availability}
The analytic proofs supporting the conclusions are contained in the main text
and Supplemental Material. The finite certificates, verification and plotting
software, tests, and data underlying the figures are openly available in the
versioned Zenodo archive \cite{ZhaoYanZenodo2026}.

\bibliographystyle{apsrev4-2}
\bibliography{reference}

\FloatBarrier
\onecolumngrid

\setcounter{section}{0}
\setcounter{equation}{0}
\setcounter{figure}{0}
\setcounter{table}{0}
\renewcommand{\thesection}{S\arabic{section}}
\renewcommand{\theequation}{S\arabic{equation}}
\renewcommand{\theHequation}{SM.\arabic{equation}}
\renewcommand{\thefigure}{S\arabic{figure}}
\renewcommand{\theHfigure}{SM.\arabic{figure}}
\renewcommand{\thetable}{S\arabic{table}}
\renewcommand{\theHtable}{SM.\arabic{table}}

\begin{center}
{\large\bf Supplemental Material for\\
``Divergence between Maximum-Probability and Maximum-Likelihood Decoding in the
Low-Noise Limit of Quantum Error Correction''}\\[6pt]
Yixin Zhao and Fei Yan
\end{center}

This Supplemental Material gives the full proofs and exact finite size
certificates underlying the main article. Secs.~\ref{sec:sm-definitions} and \ref{sec:sm-transfer} establish the
code-independent obstruction, spectrum criterion, and finite-certificate
transfer principle. Sec.~\ref{sec:sm-five} presents the transparent five-qubit
substitution mechanism before Sec.~\ref{sec:sm-four} gives the more technical four-qubit
invariant-cone proof. Secs.~\ref{sec:sm-surface} and \ref{sec:sm-toric} give analytic all-distance planar
and toric proofs. Sec.~\ref{sec:sm-hgp} proves the all-size separable
hypergraph product/bivariate bicycle (HGP/BB) theorem,
followed in Sec.~\ref{sec:sm-gross} by the exact finite-instance Gross result. Sec.~\ref{sec:sm-gap}
derives their common operational failure-gap consequence, Sec.~\ref{sec:sm-benchmark} combines
rigorous finite-noise control with an independent benchmark and distance
scaling, and Sec.~\ref{sec:sm-artifact} specifies the independently rerunnable certificate
audit. The all-distance and all-depth results follow from the analytic
arguments and exact invariant certificates developed below.

\section{Definitions and the universal obstruction}\label{sec:sm-definitions}

Along the interior Pauli ray $(p_X,p_Y,p_Z)=t(x,y,z)$, with
$x,y,z>0$, $x+y+z=1$, and $t\to0$, write
\begin{equation}
 \Pr_t(E)=t^{\operatorname{wt}(E)}\mu(E)\,[1+O(t)],
 \qquad \mu(E)>0,
 \label{eq:sm-error-expansion}
\end{equation}
where $\operatorname{wt}(E)$ is the number of physical qubits on which
$E$ acts nontrivially; $n_X(E)$, $n_Y(E)$, and $n_Z(E)$ count its three
nonidentity Pauli species; and
$\mu(E)=x^{n_X(E)}y^{n_Y(E)}z^{n_Z(E)}$ is its direction monomial. Let
$\operatorname{syn}(E)\in\{0,1\}^{n-k}$ denote the binary stabilizer
syndrome, with $[\operatorname{syn}(E)]_i=0$ ($1$) when $E$ commutes
(anticommutes) with stabilizer generator $g_i$. For each
syndrome $s$, choose one representative coset $C_s$; then $C_sL$ denotes the
compatible errors with logical label $L$. We use logical sector, coset, and
class interchangeably for this set. Define
$T_{s,L}(t)=\max_{E\in C_sL}\Pr_t(E)$ for maximum-probability (MP)
decoding and $M_{s,L}(t)=\sum_{E\in C_sL}\Pr_t(E)$ for degenerate
maximum-likelihood (MLD) decoding, with complete winner sets
$D_{\MP}(s,t)=\argmax_LT_{s,L}(t)$ and
$D_{\MLD}(s,t)=\argmax_LM_{s,L}(t)$.

\begin{smdefinition}[Strict disagreement onset]
Define the syndrome minimum weight
\begin{equation}
 w_s\equiv w_{\min}(s):=
 \min\{\operatorname{wt}(E):\operatorname{syn}(E)=s\}.
\end{equation}
The strict disagreement onset is
\begin{equation}
 m=\min\{w_s:
 D_{\MP}(s,t)\cap D_{\MLD}(s,t)=\varnothing
 \text{ for all sufficiently small }t>0\}.
\end{equation}
Equivalently, $m$ is the minimum physical weight carried by a syndrome
for
which, throughout a sufficiently small punctured interval in $t$,
\begin{equation}
D_{\MP}(s,t)\cap D_{\MLD}(s,t)=\varnothing .
 \label{eq:sm-strict}
\end{equation}
For a concatenated family we write $m_\ell$, and for a topological family
of distance $d$ we write $m_d$.
\end{smdefinition}

\begin{smtheorem}[Universal half-distance obstruction]
\label{thm:sm-lower}
Every distance-$d$ stabilizer code and every interior Pauli direction
satisfy
\begin{equation}
 m\geq\left\lceil\frac d2\right\rceil .
\end{equation}
\end{smtheorem}

\begin{proof}
Let $s_\star$ be a syndrome realizing the onset, so that
$m=w_{\min}(s_\star)$, and choose
$L_{\rm MP}$ and $L_{\rm MLD}$ from its disjoint winner sets.  The
MP-winning sector contains a weight-$m$ representative and hence has
mass $\Theta(t^m)$.  An MLD-winning sector beginning above weight $m$
would contribute only $O(t^{m+1})$ and could not win as $t\downarrow0$;
because no sector compatible with $s_\star$ begins below $m$, it also
contains a weight-$m$ representative.  Call the two representatives
$E_{\rm MP}$ and $E_{\rm MLD}$.  Their common syndrome and distinct
logical labels make $E_{\rm MP}^{\dagger}E_{\rm MLD}$ a nontrivial logical
Pauli.  Hence
\begin{equation}
 d\leq\operatorname{wt}(E_{\rm MP}^{\dagger}E_{\rm MLD})
 \leq\operatorname{wt}(E_{\rm MP})+\operatorname{wt}(E_{\rm MLD})=2m.
\end{equation}
Thus $m\geq\lceil d/2\rceil$.
\end{proof}

The set-valued definition records all MP and MLD maximizers.  A tied MP
decision therefore counts as strict disagreement precisely when its complete
winner set is disjoint from the MLD winner set.

\begin{smtheorem}[Lexicographic coset-spectrum criterion]
\label{thm:sm-spectrum}
For a syndrome $s$ and weight $w$, define the weight-resolved coset mass and
peak coefficients
\begin{align}
 P_{s,L}^{(w)}
 &=\sum_{\substack{E\in C_sL\\\operatorname{wt}(E)=w}}
 x^{n_X(E)}y^{n_Y(E)}z^{n_Z(E)},\nonumber\\
 T_{s,L}^{(w)}
 &=\max_{\substack{E\in C_sL\\\operatorname{wt}(E)=w}}
 x^{n_X(E)}y^{n_Y(E)}z^{n_Z(E)}.
 \label{eq:sm-spectrum}
\end{align}
Put $T_{s,L}=T_{s,L}^{(w_s)}$ and, for an $n$-qubit code,
\begin{equation}
 \boldsymbol P_{s,L}=
 \bigl(P_{s,L}^{(w_s)},P_{s,L}^{(w_s+1)},\ldots,P_{s,L}^{(n)}\bigr).
\end{equation}
Vectors are compared lexicographically by their first unequal component;
$\operatorname{LexMax}_L$ denotes the set of logical labels maximizing this
order. For all sufficiently small $t>0$, the MLD winners are exactly the
lexicographic maximizers of $\boldsymbol P_{s,L}$, while the MP winners
maximize $T_{s,L}$.  Hence eventual strict disagreement occurs if and
only if
\begin{equation}
 \argmax_LT_{s,L}\cap
 \operatorname{LexMax}_L\boldsymbol P_{s,L}=\varnothing.
 \label{eq:sm-spectrum-criterion}
\end{equation}
The leading-layer condition
\begin{equation}
 \argmax_LT_{s,L}\cap
 \argmax_LP_{s,L}^{(w_s)}=\varnothing
 \label{eq:sm-leading-spectrum-criterion}
\end{equation}
is sufficient, but not necessary, for eventual strict disagreement.
If every minimum-weight logical class contains exactly one
minimum-weight representative, eventual strict disagreement is
impossible for that syndrome.
\end{smtheorem}

\begin{proof}
Let $w=w_s$, write $P_{s,L}=P_{s,L}^{(w)}$, and set $\xi=t/(1-t)$.
Factoring the common probability
of $w$ nonidentity sites gives the exact finite expansion
\begin{equation}
 M_{s,L}(t)=t^w(1-t)^{n-w}
 \left[P_{s,L}+
 \sum_{j=1}^{n-w}\xi^jP_{s,L}^{(w+j)}\right].
 \label{eq:sm-exact-expansion}
\end{equation}
For two logical classes, the sign of the
difference between their brackets at sufficiently small positive $\xi$ is
the sign of its first nonzero coefficient.  This is precisely
lexicographic comparison of the finite vectors
$\boldsymbol P_{s,L}$, including any higher-order coefficient that
breaks a leading tie.

For an individual error of weight $w+j$, the same factorization leaves
$\xi^j$ times its composition monomial.  Since the direction coordinates
are at most one, no $j\geq1$ error can overtake the positive maximum
$T_s=\max_LT_{s,L}$ once $\xi<T_s$.  Among weight-$w$ errors, the common
prefactor is identical and the MP winners are exactly the logical
classes attaining $T_s$.  This proves
Eq.~(\ref{eq:sm-spectrum-criterion}).  Disjoint leading-layer winner
sets remain disjoint after lexicographic refinement, proving
Eq.~(\ref{eq:sm-leading-spectrum-criterion}).

Finally suppose every eligible logical class has one minimum-weight
representative.  Then $P_{s,L}^{(w)}=T_{s,L}$.  If $T$ has a unique
maximizer, that class also has a strictly larger first MLD coefficient.
If $T$ is tied, every possible lexicographic MLD winner belongs to the
MP tie set.  In either case the winner sets intersect, independently of
all higher coefficients.
\end{proof}

\begin{smproposition}[Half-distance localization]
\label{prop:sm-localization}
Let $h=\lceil d/2\rceil$.  If two weight-$h$ errors of the same syndrome
belong to different logical classes, their product is a nontrivial
logical Pauli of weight in $[d,2h]$.  Consequently the product has weight
$d$ for even $d$, and weight $d$ or $d+1$ for odd $d$.
\end{smproposition}

\begin{proof}
The product has zero syndrome and a nontrivial logical label, hence its
weight is at least $d$.  Subadditivity of Pauli weight gives the upper
bound $2h$.  For even $d$, $2h=d$; for odd $d$, $2h=d+1$.  No other
integer weight lies in the corresponding interval.
\end{proof}

Thus all possible half-distance reversals are determined by the
composition spectra of balanced factorizations of minimum, or for odd
distance next-to-minimum, logical operators.  This is a finite
code-independent reduction; the geometry of a code family determines
whether those spectra are pure, paired, or imbalanced.

\begin{smdefinition}[Syndrome fiber graph]
For an integer $w$, let ${\cal V}_w({\cal C})$ be the set of physical
Pauli errors of weight $w$, modulo phase.  It is partitioned by syndrome
and then by logical class, and every vertex carries its composition
$(n_X,n_Y,n_Z)$.  The syndrome fiber graph is the fibered graph
\begin{equation}
 \mathfrak F_w({\cal C})=({\cal V}_w,{\cal E}_w),
\end{equation}
where an edge $\{E,F\}\in{\cal E}_w$ is present exactly when
\begin{equation}
 \operatorname{syn}(E)=\operatorname{syn}(F),\qquad [E]\ne[F].
\end{equation}
Here $[E]$ is the logical coset within the fixed syndrome.  Every vertex
is counted once, irrespective of its graph degree.  Equivalently, every
edge factorizes a nontrivial logical Pauli $EF$ of weight at most
$2w$ into its two endpoint errors of weight $w$ each.  A fiber is the induced
decorated graph at one fixed syndrome.
\end{smdefinition}

\begin{smtheorem}[Syndrome factorization decision theorem]
\label{thm:sm-atlas}
For any finite stabilizer code and interior Pauli direction, a
leading-layer strict reversal at syndrome weight $w$ occurs exactly when
a fiber of $\mathfrak F_w$ has disjoint composition-peak and
composition-mass winner sets, where for the logical-class vertex set
${\cal V}_{w,s,L}$,
\begin{align}
 \operatorname{peak}{\cal V}_{w,s,L}
 &=\max_{E\in{\cal V}_{w,s,L}}
 x^{n_X(E)}y^{n_Y(E)}z^{n_Z(E)},\nonumber\\
 \operatorname{mass}{\cal V}_{w,s,L}
 &=\sum_{E\in{\cal V}_{w,s,L}}
 x^{n_X(E)}y^{n_Y(E)}z^{n_Z(E)}.
\end{align}
The decorated vertex partition therefore determines every leading
decision cell, while its edges identify the logical operators being
balanced.  Consequently two fiberwise isomorphisms preserving syndrome,
logical partition, and vertex composition have identical leading
strict-reversal cells at order $w$.  If every eligible logical-class
vertex set is a singleton, no eventual strict split is possible at that
syndrome weight, even if higher-order MLD coefficients break an MP tie.
\end{smtheorem}

\begin{proof}
Fix a syndrome whose minimum physical weight is $w$.  If only one
logical class has weight-$w$ representatives, both MP and MLD select it
at leading order.  Otherwise the vertex sets
${\cal V}_{w,s,L}$ are exactly the physical errors summed once in
Eq.~(\ref{eq:sm-spectrum}); hence their peak and mass are respectively
$T_{s,L}^{(w)}$ and $P_{s,L}^{(w)}$.  Vertices in two different logical
classes are connected because their product has zero syndrome and a
nontrivial logical label, so the fiber records every balanced logical
factorization without changing vertex multiplicity.
Theorem~\ref{thm:sm-spectrum} makes disjoint peak and mass winner sets
necessary and sufficient for a leading strict reversal.  A
composition-preserving fiber isomorphism preserves every vertex
monomial exactly once, hence both peak and mass and all semialgebraic
leading decision cells.  The final statement is the singleton-class
conclusion of Theorem~\ref{thm:sm-spectrum}.
\end{proof}

\begin{smdefinition}[Entropy-neutral layer and rigidity depth]
\label{def:sm-rigidity-depth}
For a code of distance $d$, put $h=\lceil d/2\rceil$.  The layer
$h+j$ is \emph{entropy neutral} on a noise region if, for every syndrome
whose minimum representative has weight $h+j$, the exact criterion in
Eq.~(\ref{eq:sm-spectrum-criterion}) leaves at least one common MP--MLD
winner throughout that region.  An exact neutrality certificate can be a
composition-preserving fiber pairing,
classwise-singleton purity, an invariant return of the concatenation
renormalization group (RG) map,
or an algebraic packing reduction.  If the first $r$ layers
$h,\ldots,h+r-1$ are certified neutral and a strict witness exists at
$h+r$, then $r$ is the \emph{rigidity depth}.  We allow
$r=\infty$ if no strict witness exists.
\end{smdefinition}

\begin{smtheorem}[Obstruction--transfer theorem]
\label{thm:sm-unified}
Let ${\cal C}_\lambda$ be a code family of distance $d_\lambda$ and put
$h_\lambda=\lceil d_\lambda/2\rceil$.
\begin{enumerate}
\item Every fiber of $\mathfrak F_w({\cal C}_\lambda)$ is edgeless for
$w<h_\lambda$.
\item Suppose an exact family map sends a decorated same-syndrome pair
at weight $w_\lambda$ to such a pair at $w_{\lambda'}$, preserves its
logical separation, and maps a strict peak--mass decision cell into the
next strict cell.  Then a disagreement witness transfers from $\lambda$
to $\lambda'$.
\item If the layers
$h_\lambda,\ldots,h_\lambda+r-1$ have exact entropy-neutral certificates,
while a fiber at $h_\lambda+r$ has a strict peak--mass reversal, then
\begin{equation}
 m_\lambda=h_\lambda+r.
 \label{eq:sm-rigidity-depth}
\end{equation}
\end{enumerate}
The family map may be a composition-preserving geometric isomorphism or
an exact leading-term RG return.  Classwise-singleton purity also certifies
entropy neutrality.
\end{smtheorem}

\begin{proof}
An edge below $h_\lambda$ would factor a nontrivial logical into two
weight-$w$ errors and give $d_\lambda\leq2w<2h_\lambda$, contradicting
the even case directly and the odd case because
$w\leq h_\lambda-1=(d_\lambda-1)/2$.  Part 2 follows by applying the
exact map to the two endpoints and to their peak and mass inequalities.
For part 3, the exact neutral layer certificates exclude all strict splits
from $h_\lambda$ through $h_\lambda+r-1$, the imbalanced fiber supplies the
upper witness at $h_\lambda+r$, and Theorem~\ref{thm:sm-lower} excludes all
weights below $h_\lambda$.  Hence the two bounds meet.
\end{proof}

At $w=h=\lceil d/2\rceil$,
Proposition~\ref{prop:sm-localization} reduces the entire graph to
balanced factorizations of weight-$d$ logicals, and also weight $d+1$
logicals for odd $d$.  More generally, increasing $w$ constructs a
finite hierarchy of near-minimal logical factorizations; the
first non-neutral member fixes the rigidity depth when it contains a strict
witness.  In general an exact neutrality certificate may require the
lexicographic higher layers of Theorem~\ref{thm:sm-spectrum}; leading
composition matching alone is insufficient.  In the six realizations proved
here the explicit RG, boundary-path, loop, purity, or packing certificates establish
every required neutral layer and the first strict witness, so the reported
values $r=0,1,2$ and $m=h+r$ are exact.

\section{Finite-certificate transfer principle and exact recursion invariants}\label{sec:sm-transfer}

\begin{smdefinition}[Leading term]
For a nonzero decoder message,
\begin{equation}
 g(f)=(\nu,c)
 \quad\Longleftrightarrow\quad
 f(t)=ct^\nu+O(t^{\nu+1}),\qquad c>0.
 \label{eq:sm-germ}
\end{equation}
At unequal exponents a sum retains the smaller $\nu$; at equal
exponent MLD adds leading coefficients while MP retains their maximum.
Products add valuations and multiply coefficients.  Contracting a fixed
seed stabilizer tensor therefore induces exact finite maps on vectors of
leading terms. The exponent-only part is the min-plus rule: ordinary products
add exponents, while ordinary sums select their minimum.
\end{smdefinition}

A decorated finite certificate contains:
\begin{enumerate}
\item the seed stabilizer tensor, logical-label convention, and rational
noise cell;
\item a finite phase set ${\cal P}$ and decorated support-state set
${\cal S}$;
\item for every state, two physical representatives, their common
syndrome, actual logical labels, child decomposition, Pauli support, and
overlap data;
\item an exact transition
$\sigma:{\cal P}\times{\cal S}\rightarrow{\cal P}\times{\cal S}$;
\item rational polyhedral exponent cells $K_{p,s}$ and positive
coefficient cells $C_{p,s}$ that fix every selected branch and the
opposite strict decoder winners;
\item exact return maps $V_{p,s}$ and $A_{p,s}$;
\item a finite entrance reconstructed from explicit physical
representatives; and
\item symbolic identities or nonnegative rational Farkas multipliers,
coefficients expressing each pulled-back inequality as a nonnegative linear
combination of the defining inequalities, for all return inclusions.
\end{enumerate}

\begin{smtheorem}[Invariant-cell criterion]
\label{thm:sm-certificate}
Suppose every stored transition agrees with direct contraction of the
declared seed tensor and
\begin{equation}
 V_{p,s}(K_{p,s})\subseteq K_{\sigma(p,s)},\qquad
 A_{p,s}(C_{p,s})\subseteq C_{\sigma(p,s)}
 \label{eq:sm-inclusions}
\end{equation}
for every phase-state pair.  Suppose all intermediate selected branches
are strict and the decorated transition preserves a same-syndrome pair
whose product is a minimum logical operator of weight $d_\ell$ and whose
endpoint weights are $\lceil d_\ell/2\rceil$, allowing complementary
intermediate weights
$(\lfloor d_\ell/2\rfloor,\lceil d_\ell/2\rceil)$.  If an explicit
witness enters these cells at level $L$, then
\begin{equation}
 m_\ell=\left\lceil\frac{d_\ell}{2}\right\rceil,\qquad \ell\geq L.
\end{equation}
Independent checks below $L$ extend the equality to every level.
\end{smtheorem}

\begin{proof}
At the entrance level, direct contraction verifies a same-syndrome
half-distance pair with a minimum-logical product and disjoint strict
MLD/MP winners.  Assume the statement in phase-state $(p,s)$ at level
$\ell$.  The stored seed identity constructs the next complementary
pair and preserves its syndrome, actual logical labels, product, endpoint
weights, and overlap.  The first inclusion in
Eq.~(\ref{eq:sm-inclusions}) preserves every min-plus branch used in the
construction.  The second preserves every leading-coefficient
inequality that makes the decoder decisions strict and disjoint.  The
next state is therefore another half-distance witness.  Induction gives
the upper bound at all later levels, and
Theorem~\ref{thm:sm-lower} gives equality.
\end{proof}

For a cone with row vectors $g_i$ and return matrix $M$, every inclusion
used here has an exact certificate
\begin{equation}
 g_iM=\sum_j\lambda_{ij}g_j,\qquad
 \lambda_{ij}\in\mathbb Q_{\geq0}.
 \label{eq:sm-farkas}
\end{equation}
This separates discovery from proof.  Search may propose a state or
cone, but verification recomputes seed contractions and checks
Eq.~(\ref{eq:sm-farkas}) over integers and rationals.  A stored success
flag is never a proof input.

\section{Cyclic \texorpdfstring{$[[5,1,3]]$}{[[5,1,3]]} family}\label{sec:sm-five}

Use
\begin{equation}
 S_5=\langle XZZXI,IXZZX,XIXZZ,ZXIXZ\rangle,\qquad
 \overline X=XXXXX,\quad\overline Z=ZZZZZ.
 \label{eq:sm-five-code}
\end{equation}
Its $\ell$-fold self-concatenation has $d_\ell=3^\ell$.
Complete level-two enumeration compresses the $4^5$ child patterns into
56 state--actual-label--weight types.  Among the 19,998 ordered type
tuples of total weight five, exactly 18 yield strict disagreement.  Every
backtracked pair satisfies
\begin{equation}
 (w_L,w_R,|\operatorname{supp}L\cap\operatorname{supp}R|,
 \operatorname{wt}(LR))=(5,5,1,9).
 \label{eq:sm-five-entrance}
\end{equation}
Expansion gives 360 exact level-three pairs and all 12 oriented pairs of
distinct MLD and MP winners.

For logical target $a\in\{1,2,3\}$, let $U_a$ denote two half-distance
representatives with one overlap.  Resolving the overlap to either side
gives $R_a^L,R_a^R$.  Together with $I$, the decorated alphabet is
\begin{equation}
 I,\ U_1,U_2,U_3,\
 R_1^L,R_1^R,R_2^L,R_2^R,R_3^L,R_3^R.
\end{equation}
Let $b(1)=3$, $b(2)=1$, and $b(3)=2$.  The five-child substitution is
\begin{equation}
 U_a'=(I,I,R_b^L,U_a,R_b^R),
 \label{eq:sm-five-substitution}
\end{equation}
with $(R_a^L)'$ and $(R_a^R)'$ obtained by replacing the fourth child by
the corresponding resolved state.

\begin{table}[h]
\centering
\begin{tabular}{c ccc ccc}
\toprule
&\multicolumn{3}{c}{odd source phase}
&\multicolumn{3}{c}{even source phase}\\
$a$ & $U_a$ & $R_a^L$ & $R_a^R$
& $U_a$ & $R_a^L$ & $R_a^R$\\
\midrule
1 & $(1,0)$ & $(3,2)$ & $(2,3)$
  & $(2,3)$ & $(0,1)$ & $(1,0)$\\
2 & $(2,0)$ & $(1,3)$ & $(3,1)$
  & $(3,1)$ & $(0,2)$ & $(2,0)$\\
3 & $(3,0)$ & $(2,1)$ & $(1,2)$
  & $(1,2)$ & $(0,3)$ & $(3,0)$\\
\bottomrule
\end{tabular}
\caption{Actual logical-label pairs, with $0=I,1=X,2=Y,3=Z$.  The
identity state is $(0,0)$ in both phases.}
\label{tab:sm-five-labels}
\end{table}

\begin{smlemma}[Odd-overlap induction]
\label{lem:sm-five-support}
If $U_a$ has endpoint weight
$h_\ell=(3^\ell+1)/2$, one overlap, and product weight $3^\ell$, then
Eq.~(\ref{eq:sm-five-substitution}) has the same parent syndrome and
logical target, endpoint weight $h_{\ell+1}$, one overlap, and product
weight $3^{\ell+1}$.
\end{smlemma}

\begin{proof}
Each child pair has a common syndrome, so applying the same outer seed
transition gives the same parent syndrome.  Direct contraction of
Eq.~(\ref{eq:sm-five-code}) maps the five child actual labels to the
corresponding entry of Table~\ref{tab:sm-five-labels}.  This is a finite
$2\times10$ seed identity; cyclic relabeling gives the displayed
predecessor rule.  The resolved states have weights
$(h_\ell,h_\ell-1)$ and $(h_\ell-1,h_\ell)$.  Either endpoint therefore
has weight
\begin{equation}
 h_\ell+(h_\ell-1)+h_\ell
 =3h_\ell-1=\frac{3^{\ell+1}+1}{2}.
\end{equation}
Only the unresolved fourth child retains an overlap.  The product has
three disjoint minimum-logical child products and hence weight
$3^{\ell+1}$.
\end{proof}

The exponent recursion has two parity phases.  With
\begin{equation}
 q_\ell=\frac{3^{\ell-1}-3}{2},\qquad
 q_{\ell+1}=3q_\ell+3,\qquad q_3=3,
\end{equation}
every order lies in
\begin{equation}
 \{0,1,q_\ell,q_\ell+1,q_\ell+2,3^\ell\}.
\end{equation}
For each phase, decoder, nonidentity decorated state, and logical
component, exact seed contraction selects one assignment for all integer
$q\geq3$.  The unquotiented audit therefore checks
\begin{equation}
 2\ {\rm phases}\times2\ {\rm decoders}\times
 9\ {\rm states}\times4\ {\rm components}=144
\end{equation}
unique branches and 36 nonidentity message returns.  Cyclic covariance
reduces the human-readable mechanism to the three modes
$U_a,R_a^L,R_a^R$, while the certificate retains every target.

At odd levels, the main state $U_2$ has competing order-zero classes
$(0,2)$ and exact ratios
\begin{equation}
 \frac{c_{\MLD,2}}{c_{\MLD,0}}
 =\frac{2yz}{x^2+y^2},\qquad
 \frac{c_{\MP,0}}{c_{\MP,2}}=\frac yz.
 \label{eq:sm-five-ratios}
\end{equation}
At even levels the same ratios occur on classes $(1,3)$ with the
corresponding orientation.  The symbolic return treats both
equal-exponent pairs as arbitrary positive scales and preserves
Eq.~(\ref{eq:sm-five-ratios}) exactly.

\begin{smproposition}[Analytic five-qubit phase]
\label{prop:sm-five-box}
All entrance, branch, and decision inequalities above are strict on the
open semialgebraic region
\begin{equation}
 y>z>x>0,\qquad xy^2<z^3,\qquad
 x^2+y^2<2yz.
 \label{eq:sm-five-region}
\end{equation}
\end{smproposition}

\begin{proof}
After symbolic duplicate removal, there are 19 rational inequalities.
Factoring their numerators and denominators over the integers leaves
positive monomials and sums, together with only
\begin{equation}
 y-z,\quad z-x,\quad y-x,\quad z^2-xy,\quad
 z^3-xy^2,\quad 2yz-x^2-y^2.
\end{equation}
The first three have the stated signs in
Eq.~(\ref{eq:sm-five-region}).  The fourth follows because
$xy^2<z^3$ and $y>z$ imply $xy<z^3/y<z^2$; the final two are primitive
inequalities of the region.  Every denominator is a product of positive
coordinates and positive sums of squares.  Hence all 19 rational
functions are strictly positive.  The released checker independently
factors each recorded certificate inequality and verifies this implication
symbolically.
\end{proof}

The phase contains the closed rational rectangle
\begin{equation}
 \frac{17}{250}\leq x\leq\frac{33}{250},\qquad
 \frac{117}{250}\leq y\leq\frac{133}{250},\qquad z=1-x-y.
\end{equation}
The four worst-case strict margins are respectively
$1/250$, $51/250$, $8967/15625000$, and $439/31250$.

The level-two witness gives $m_2=5$; a level-three weight-14 pair enters
the enumeration.  Lemma~\ref{lem:sm-five-support}, the 144 branch identities,
the 36 coefficient returns, and Proposition~\ref{prop:sm-five-box}
propagate strict opposite winners to all later levels.
Theorem~\ref{thm:sm-lower} then proves
\begin{equation}
 m_\ell=\frac{3^\ell+1}{2},\qquad \ell\geq2.
\end{equation}

\section{Noise-tailored \texorpdfstring{$[[4,1,2]]$}{[[4,1,2]]} family}\label{sec:sm-four}

Take
\begin{equation}
 S_4=\langle IYIX,ZIZI,YZXY\rangle,\qquad
 \overline X=YIXI,\quad\overline Z=YIYX.
\end{equation}
This is locally Clifford and qubit-permutation equivalent to
$\langle XXXX,ZZII,IIZZ\rangle$.  Since
$(ZZII)(IIZZ)=ZZZZ$, this conventional stabilizer is the standard
$[[4,2,2]]$ group $\langle XXXX,ZZZZ\rangle$ with the additional
commuting logical constraint $ZZII$; it is therefore the
$[[4,1,2]]$ subcode obtained by fixing that logical constraint.  One explicit
equivalence first permutes the old qubits to $(0,2,1,3)$ and then applies
the phase-free axis maps
\begin{equation}
\begin{array}{c|ccc}
j&X&Y&Z\\ \hline
0&X&Y&Z\\
1&Y&X&Z\\
2&X&Z&Y\\
3&Z&Y&X
\end{array}.
\end{equation}

Let $(0,a_\ell,b_\ell,c_\ell)$ be the minimum weights of the four
zero-syndrome logical classes.  The seed vector is $(3,2,2)$.  Complete
min-plus contraction gives
\begin{align}
 a'&=\min\{3c,b+2c,2b+c,3b,a+b+c,3a\},\nonumber\\
 b'&=\min\{b+c,a+c,a+b\},\nonumber\\
 c'&=\min\{b+3c,3b+c,a+c,a+b\}.
 \label{eq:sm-four-minplus}
\end{align}
On the cone $b\leq c\leq a\leq2b$, direct comparison of every entry in
Eq.~(\ref{eq:sm-four-minplus}) selects
\begin{equation}
 a_{\ell+1}=3b_\ell,\qquad
 b_{\ell+1}=b_\ell+c_\ell,\qquad
 c_{\ell+1}=a_\ell+b_\ell.
\end{equation}
Indeed, the first output is bounded below by $3b$, the second by $b+c$,
and the third by $a+b$, with equality at the displayed branches.  The
image obeys
\begin{equation}
 b+c\leq a+b\leq3b\leq2(b+c),
\end{equation}
so the cone is invariant.  Since $b_\ell$ is the minimum of the three
nontrivial zero-syndrome logical weights, it is the ordinary quantum
distance $d_\ell$.  Eliminating $a_\ell,c_\ell$ gives
\begin{equation}
 d_{\ell+3}=d_{\ell+2}+d_{\ell+1}+3d_\ell,\qquad
 (d_1,d_2,d_3)=(2,4,9).
 \label{eq:sm-four-distance}
\end{equation}

The parity vector is periodic:
\begin{equation}
 111\longrightarrow100\longrightarrow001
 \longrightarrow010\longrightarrow111.
 \label{eq:sm-four-phases}
\end{equation}
For every phase and logical target the certificate stores a balanced pair
$B_j$, complementary resolutions $C_j^L,C_j^R$ when the pair overlaps,
and a four-child support recipe.  Writing
$d_j=2k_j+\epsilon_j$, direct symbolic counting in all 12 phase-target
recipes gives
\begin{equation}
 w_L=w_R=\left\lceil\frac{d_j'}2\right\rceil,\quad
 \operatorname{wt}(E_LE_R)=d_j',\quad
 |\operatorname{supp}E_L\cap\operatorname{supp}E_R|
 =d_j'\bmod2.
 \label{eq:sm-four-support}
\end{equation}

The coefficient state has left/right and MLD/MP copies.  In logarithmic
ratio coordinates
\begin{equation}
 x=(b_1,b_2,b_3,l_1,r_1,l_2,r_2,l_3,r_3),
\end{equation}
the absorbing mixed cone has eight facets
\begin{equation}
\begin{gathered}
 b_1,b_2,b_3>0,\qquad
 l_1+r_1,l_2+r_2,l_3+r_3>0,\\
 l_2+r_3>0,\qquad -l_1+l_2>0.
 \label{eq:sm-four-facets}
\end{gathered}
\end{equation}
The first three orient balanced targets, the next three orient products
of complementary resolutions, and the last two fix the required mixed
cross orientation.  The positive entrance cone reaches this chamber
after one four-level return, and all eight pullbacks have nonnegative
rational Farkas decompositions.

For exponent dynamics use
$v=(a,b,c,q_1,q_2,q_3,h)^T$ with $h=1$.  The four-level return is
\begin{equation}
v'=
\begin{pmatrix}
3&18&6&0&0&0&0\\
2&11&6&0&0&0&0\\
4&12&5&0&0&0&0\\
1&3&2&0&1&0&-2\\
2&9&3&0&0&1&-1\\
1&8&3&1&0&0&1\\
0&0&0&0&0&0&1
\end{pmatrix}v.
\label{eq:sm-four-return}
\end{equation}
All selected-minimum comparisons reduce, after duplicate collection, to
the following row families.
\begin{table}[h]
\centering
\begin{tabular}{lrr}
\toprule
Row family & Unique rows & Expanded comparisons\\
\midrule
balanced selected-minimum chamber & 184 & 332\\
complementary selected-minimum chamber & 90 & 168\\
base distance nonnegativity & 4 & 4\\
\midrule
total & 278 & 504\\
\bottomrule
\end{tabular}
\caption{Compressed exponent-certificate families.}
\label{tab:sm-four-rows}
\end{table}
Their phase distribution is $60,56,69,89$ in phases
$111,100,001,010$, plus four phase-independent base rows.  For each
unique row $g_i$, the released certificate stores its phase, represented
linear comparison, multiplicity, return row $g_iM$, and exact
nonnegative multipliers in Eq.~(\ref{eq:sm-farkas}).  The depth-24 state
lies strictly inside the cone with minimum integer slack one.

\begin{smproposition}[Finite entrance and absorption]
Exact decoding through the finite entrance reaches the positive
coefficient cone.  The exact bridge through level 24 enters the mixed
coefficient chamber and the 278-row exponent cone.  All intermediate
branches are unique and Eq.~(\ref{eq:sm-four-support}) holds in every
phase.
\end{smproposition}

\begin{proof}
Coefficient maps are integer-linear after logarithmic ratios and
exponent maps are integer affine in the lifted coordinate $h=1$.
Entrance and bridge states are evaluated with exact integers.  The eight
coefficient and 278 exponent pullbacks are checked with rational Farkas
multipliers.
\end{proof}

On
\begin{equation}
 \frac{99}{1000}\leq x\leq\frac{101}{1000},\qquad
 \frac{499}{1000}\leq y\leq\frac{501}{1000},\qquad z=1-x-y,
\end{equation}
fixed half-distance witnesses at levels $1$--$11$ are strict on a
$2\times2$ cover.  At level 12, 28 ordinary entrance facets hold on one
outward-rounded interval box; the four remaining cross-orientation
facets are reduced, on a $4\times2$ cover, by rigorous logarithmic
derivative intervals to exact rational extremal-corner comparisons.
The invariant cells then propagate the uniform entrance.  Combining this
upper construction with
Theorem~\ref{thm:sm-lower} proves
\begin{equation}
 m_\ell=\left\lceil\frac{d_\ell}{2}\right\rceil
\end{equation}
for every depth.

\section{Rotated planar surface code: boundary-path multiplicity}\label{sec:sm-surface}

This section proves the statement analytically for every odd distance.
Finite trellis evaluation (dynamic-programming enumeration) checks selected
distances. To distinguish the
distance parameter from the rigidity depth, write $d=2\rho+1$ and
$h=\rho+1$.  Index data qubits by
$(i,j)\in\{0,\ldots,d-1\}^2$.  To make the boundary convention
unambiguous, the $X$ checks are
\begin{align}
 B^{X,\mathrm t}_j&=X_{0,2j}X_{0,2j+1},&
 B^{X,\mathrm b}_j&=X_{d-1,2j+1}X_{d-1,2j+2},\nonumber\\
 A^X_{i,j}&=\prod_{\epsilon,\eta=0}^1X_{i+\epsilon,j+\eta},
 &&i+j=1\pmod2,
\label{eq:sm-surface-xchecks}
\end{align}
and the $Z$ checks are
\begin{align}
 B^{Z,\mathrm l}_i&=Z_{2i+1,0}Z_{2i+2,0},&
 B^{Z,\mathrm r}_i&=Z_{2i,d-1}Z_{2i+1,d-1},\nonumber\\
 A^Z_{i,j}&=\prod_{\epsilon,\eta=0}^1Z_{i+\epsilon,j+\eta},
 &&i+j=0\pmod2.
\label{eq:sm-surface-zchecks}
\end{align}
Here bulk indices lie in $\{0,\ldots,d-2\}$ and boundary indices in
$\{0,\ldots,\rho-1\}$.  This gives the standard
$[[d^2,1,d]]$ rotated planar code with
\begin{equation}
 \overline X=\prod_{i=0}^{d-1}X_{i,0},\qquad
 \overline Z=\prod_{j=0}^{d-1}Z_{0,j}.
\end{equation}
Define
\begin{equation}
 E_d=Z_{(1,0)}X_{(0,0)}
 \prod_{i=3}^{\rho+1}X_{(i,0)}.
 \label{eq:sm-surface-error}
\end{equation}
Both $E_d$ and $E_d\overline X$ have weight $h$ and the same syndrome.
Set
\begin{equation}
 N=\begin{cases}\rho,&\rho\ {\rm odd},\\\rho-1,&\rho\ {\rm even},\end{cases}
\qquad a=\rho-N+1,\qquad J=(N-1)/2.
\end{equation}

\begin{figure}[h]
\centering
\begin{minipage}[t]{0.52\textwidth}
\vspace{0pt}
\raggedright\small (a)\par
\centering
\includegraphics[width=0.78\linewidth]{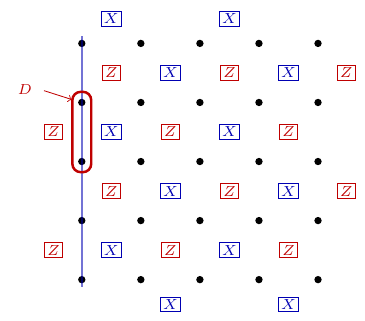}
\end{minipage}\hfill
\begin{minipage}[t]{0.44\textwidth}
\vspace{0pt}
\raggedright\small (b)\par
\centering
\includegraphics[width=0.95\linewidth]{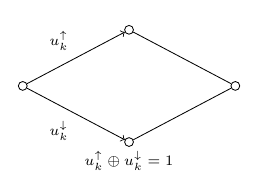}
\end{minipage}
\caption{Explicit planar convention and boundary-path cell.  (a) A
$5\times5$ patch of
Eqs.~(\ref{eq:sm-surface-xchecks})--(\ref{eq:sm-surface-zchecks});
the pattern repeats for arbitrary odd $d$.  The blue column is
$\overline X$, and the red pair is the $v$ boundary cut $D$.  (b)
Equality in the diagonal cuts contracts each two-row checkerboard cell
to the displayed binary layer.}
\label{fig:sm-surface-convention}
\end{figure}

\begin{smlemma}[Boundary-path enumerators]
\label{lem:sm-surface}
Exactly two logical classes $A$ and $B$ occur at physical weight $h$,
with composition enumerators
\begin{align}
 P_{A,d}&=2\sum_{j=0}^{J}\binom Nj
 x^{a+2j}y^{N-2j},\nonumber\\
 P_{B,d}&=4\sum_{j=0}^{J}
 \left(\binom Nj-\binom N{j-1}\right)
 x^{a+2j}y^{N-1-2j}z.
 \label{eq:sm-surface-spectra}
\end{align}
\end{smlemma}

\begin{proof}
Write a Pauli error as a pair $(u,v)$, where $u$ is its $X/Y$ chain in
the $Z$-check decoding graph and $v$ its $Z/Y$ chain in the dual
$X$-check graph.  Its physical weight is the joint support
$|u\cup v|$, not $|u|+|v|$.  Label the alternating checks and rough
boundary vertices along the first data column by
$V_0,\ldots,V_d$, so qubit $(i,0)$ is the graph edge
$V_iV_{i+1}$.  More generally a data qubit is the edge joining its two
incident checks in Eqs.~(\ref{eq:sm-surface-xchecks}) or
(\ref{eq:sm-surface-zchecks}); a missing check is replaced by the
corresponding boundary vertex.  Thus the two decoding graphs, their
boundary vertices, and every parity equation are fixed explicitly by
Eqs.~(\ref{eq:sm-surface-xchecks})--(\ref{eq:sm-surface-zchecks}).

\emph{Step 1: exact sector bounds.}
The $u$ terminals of $E_d$ are
$\{V_0,V_1,V_3,V_{\rho+2}\}$; multiplication by $\overline X$ replaces
$V_0$ by $V_d$.  Give every $Z$-graph vertex its diagonal height along
the first column.  A graph edge changes this height by at most one.
For each integer $k$, the edge boundary between heights below and above
$k$ is a diagonal cut, and distinct graph edges cross at most one such
cut.  Hence a chain joining $V_i$ to $V_j$ has length at least
$|i-j|$.  Evaluating the three terminal pairings gives
\begin{equation}
\begin{array}{c|cc}
 &0&1\\ \hline
\min |u|\text{ in logical sector}&\rho+1&\rho\\
\min |v|\text{ in logical sector}&2\rho&1 .
\end{array}
\label{eq:sm-surface-sector-bounds}
\end{equation}
For the $u$ sectors the attaining pairings are
\begin{align}
\alpha=1:&\quad(V_0,V_1),(V_3,V_{\rho+2}),\nonumber\\
\alpha=0:&\quad(V_1,V_3),(V_{\rho+2},V_d).
\label{eq:sm-surface-pairings}
\end{align}
The first-column paths attain every bound.  The $v$ syndrome is the
single $X$ check $A^X_{1,0}$ adjacent to the left boundary.  Its two
boundary edges are $(1,0)$ and $(2,0)$, giving the length-one
$\beta=1$ sector; the opposite sector crosses all $d-1=2\rho$ vertical
diagonal cuts.  For $\rho\geq2$, the
$v$-sector of length $2\rho$ cannot occur at physical weight $h=\rho+1$.
For $d=3$, row reduction of the eight local check equations gives the
complete weight-two table
\begin{equation}
\begin{array}{c|c|c}
\text{class}&(n_X,n_Y,n_Z)&\text{multiplicity}\\ \hline
01&(1,1,0)&2\\
11&(1,0,1)&4\\
00,10&-&0
\end{array}
\label{eq:sm-surface-d3-table}
\end{equation}
which is also the $N=a=1$ instance of the formulas below.  It remains
for general $\rho$ to classify the two sectors
\[
 A=(0,1),\qquad B=(1,1).
\]

\emph{Step 2: equality graph and joint-support peeling.}
For class $A$, the $h$ diagonal cuts that prove $|u|\geq h$ already
charge $h$ pairwise disjoint physical edge sets.  Equality in the Pauli
union therefore forces exactly one crossing of each cut, no edge outside
their union, and $v\subseteq u$.  Thus $u$ is one of the geodesics in
Eq.~(\ref{eq:sm-surface-pairings}); a longer $u$ cannot be rescued by
$Y$ overlap.

Class $B$ needs one additional cut because its separate $u$ minimum is
$\rho=h-1$.  Choose the $\rho$ disjoint diagonal cuts used by its two
pairings in Eq.~(\ref{eq:sm-surface-pairings}), and in the $X$ decoding
graph take the boundary cut
\begin{equation}
 D=\{(1,0),(2,0)\}.
\label{eq:sm-surface-boundary-cut}
\end{equation}
The unique $v$-syndrome terminal lies on the check side of $D$, so
\begin{equation}
 |v\cap D|=1\pmod2.
\end{equation}
 For class $B$, the attaining $u$ pairing is
 $(V_0,V_1),(V_3,V_{\rho+2})$.  Its first edge is $(0,0)$ and its remaining
 diagonal cuts start below $V_3$; their physical supports therefore
 avoid the two-edge set $D$.  Hence these
$\rho+1=h$ edge-disjoint cuts calibrate $|u\cup v|$ directly.  Equality
contains one edge in every selected cut and no edge elsewhere.  The
$Z$-check terminal pairing itself puts no $u$ edge in $D$, so precisely one
$v$ edge lies outside $u$.  This proves, before any enumeration, that
additional $Y$ overlap cannot compensate a nongeodesic component.

We now solve the equality parity equations.  This reduction is
constructive: for each attaining terminal pair $(p,q)$ retain a
decoding-graph edge $ab$ precisely when
\begin{equation}
 \operatorname{dist}(p,a)+1+\operatorname{dist}(b,q)
 =\operatorname{dist}(p,q)
\label{eq:sm-surface-geodesic-test}
\end{equation}
or with $a,b$ interchanged, take the union for the two pairs, and delete
every edge outside the calibrated cuts.  The graph distance is computed
in the decoding graph fixed above, while the diagonal cuts certify its
terminal value.  Hence
Eq.~(\ref{eq:sm-surface-geodesic-test}) is an explicit test on the
checks in Eqs.~(\ref{eq:sm-surface-xchecks})--(\ref{eq:sm-surface-zchecks}).
On every retained check vertex $q$
of the two explicitly defined decoding graphs they are
\begin{equation}
 \bigoplus_{e\ni q}u_e=s_q^Z,\qquad
 \bigoplus_{e\ni q}v_e=s_q^X.
\label{eq:sm-surface-parity-equations}
\end{equation}
Contract forced degree-two segments in the graph just defined.  Direct
substitution of the alternating bulk and boundary checks gives the
following incidence audit:
\begin{equation}
\begin{array}{c|c|c|c}
\text{location}&\text{incoming}&\text{continuations}&\text{condition}\\ \hline
\text{bulk two-row cell}&1&2&\uparrow\text{ or }\downarrow\\
\text{unfolded smooth edge}&1&2&H\in\mathbb Z\\
\text{rough edge}&1&1\text{ or }2&H\geq0\\
\text{last cell, }\rho\text{ odd}&1&2&N=\rho,\ a=1\\
\text{last cell, }\rho\text{ even}&1&1&N=\rho-1,\ a=2
\end{array}
\label{eq:sm-surface-incidence-audit}
\end{equation}
The first row repeats after translating by two data rows; the last two
rows exhaust the two checkerboard parities.  Thus every uncontracted
vertex is either one of the displayed boundary vertices or a translate
of the single cell in Fig.~\ref{fig:sm-surface-convention}.  At layer
$k$ exactly two geodesic edges remain;
write their bits as $u_k^\uparrow,u_k^\downarrow$.  The first equation
in Eq.~(\ref{eq:sm-surface-parity-equations}) becomes
\begin{equation}
 u_k^\uparrow\oplus u_k^\downarrow=1,\qquad
 H_k=H_{k-1}+u_k^\uparrow-u_k^\downarrow .
\label{eq:sm-surface-walk-equations}
\end{equation}
Thus each layer is one up/down step.  Unfolding the smooth boundary gives
an unrestricted walk for $A$; at the rough boundary the physical
half-plane imposes $H_k\geq0$ for $B$.  The two-row checkerboard period
leaves
\begin{equation}
 N=\begin{cases}\rho,&\rho\ {\rm odd},\\\rho-1,&\rho\ {\rm even},\end{cases}
\end{equation}
binary layers and $a=\rho-N+1$ forced boundary $X$ edges.
The complete Boolean table for the local $u$ equation is
\begin{equation}
\begin{array}{cc|c|c}
u_k^\uparrow&u_k^\downarrow&
u_k^\uparrow\oplus u_k^\downarrow& H_k-H_{k-1}\\ \hline
0&0&0&\text{forbidden}\\
0&1&1&-1\\
1&0&1&+1\\
1&1&0&\text{forbidden}
\end{array}
\label{eq:sm-surface-transition-table}
\end{equation}
In physical coordinates, advancing one contracted layer means advancing
two checkerboard rows; the two retained edges leave the same $Z$ check
through the upper-right or lower-right data qubit in
Fig.~\ref{fig:sm-surface-convention}.  Eqs.~(\ref{eq:sm-surface-geodesic-test})
and~(\ref{eq:sm-surface-incidence-audit}) show that repeating this table
generates the entire equality graph.  At a boundary only the row that
remains inside the patch is retained, exactly $H_k\geq0$.

 After a choice in Eq.~(\ref{eq:sm-surface-walk-equations}), every
 interior $X$ check has degree two in the remaining support.  If its two
 incident retained bits are $v_{\rm in},v_{\rm out}$, the second equation
 in Eq.~(\ref{eq:sm-surface-parity-equations}) reads
 $v_{\rm out}=v_{\rm in}\oplus s_q^X$.  Its exhaustive truth table is
\begin{equation}
\begin{array}{cc|c}
v_{\rm in}&s_q^X&v_{\rm out}\\ \hline
0&0&0\\0&1&1\\1&0&1\\1&1&0
\end{array}.
\label{eq:sm-surface-v-table}
\end{equation}
All peeled interior checks have $s_q^X=0$; the unique terminal check
has $s_q^X=1$.  Thus the entering boundary bit fixes every $v$ bit.
Substitution on one
two-row cell gives the complete equality table
\begin{equation}
\begin{array}{c|cc|c|c}
\text{motif}&u\text{ bits}&v\text{ bits}&\text{Pauli}&
|u\cup v|\\ \hline
A\text{ boundary}&1^a&0^a&X^a&a\\
B\text{ boundary}&1^a;0&0^a;1&X^aZ&a+1\\
\text{cancelled up--down pair}&11&00&XX&2\\
\text{unmatched walk step}&1&1&Y&1
\end{array}
\label{eq:sm-surface-local-table}
\end{equation}
The semicolon separates the unique edge of $D$.  These are exhaustive:
apply the canonical left-to-right stack cancellation to any odd-length
up/down word.  It pairs every minority step with an opposite step and
leaves only unmatched majority steps, while
Eq.~(\ref{eq:sm-surface-parity-equations}) leaves only the displayed
$v$ assignment on each peeled cell.  Equivalently, an innermost matched
pair is an adjacent up--down or down--up detour; deleting its two-row
cell contributes $XX$ and leaves the same boundary bits.  Induction
removes all matched pairs, after which each uniform unmatched cell is
$Y$.  The two possible initial boundary
bits give the two $v$ completions in class $A$; in class $B$ there are
two top-boundary orientations and two choices of the edge in $D$.
The bases $\rho=1,2$ are exactly the first two rows of the same table.
Peeling one two-row cell preserves
Eqs.~(\ref{eq:sm-surface-walk-equations})--(\ref{eq:sm-surface-local-table}),
which proves the equality characterization for every odd $d$.

\emph{Step 3: enumerate the equality walks.}
For an unrestricted length-$N$ word with $j$ minority steps, the
$j$ paired turns contribute $2j$ factors of $X$, its $N-2j$ unmatched
steps contribute $Y$, and the boundary gadget contributes $X^a$.
There are $2\binom Nj$ such words, giving $P_{A,d}$.

For a nonnegative word with $j$ down steps, the reflection principle
gives $\binom Nj-\binom N{j-1}$ words.  Its paired turns again contribute
$X^{2j}$; one unmatched step terminates at the unique outside $Z$ edge
and the remaining $N-1-2j$ contribute $Y$.  Two rough-boundary
orientations and two forced $v$ completions give the factor four and
$P_{B,d}$.  This proves Eq.~(\ref{eq:sm-surface-spectra}) and exhausts
all physical weight-$h$ representatives.
\end{proof}

For $y>x$, the $j=0$ monomial is uniquely MP-optimal in each class, and
their ratio is $y/z$; thus $y>z$ makes MP choose $A$.  Let
$u=x^2/y^2$ and
\begin{align}
 S_N(u)&=\sum_{j=0}^{q}\binom Nj u^j,\nonumber\\
 T_N(u)&=\sum_{j=0}^{q}
 \left(\binom Nj-\binom N{j-1}\right)u^j.
\end{align}
The identity and its consequence can be written together as
\begin{align}
 T_N(u)&=(1-u)S_N(u)+\binom Nq u^{q+1}
 >(1-u)S_N(u),\nonumber\\
 \frac{P_{B,d}}{P_{A,d}}
 &=\frac{2zT_N(u)}{yS_N(u)}
 >\frac{2z(1-u)}y.
\end{align}
Consequently, in the main article's surface region MLD uniquely chooses $B$.
The strict witness has order $h$, and
Theorem~\ref{thm:sm-lower} proves
$m_d^{\rm planar}=(d+1)/2$.

\section{Odd square toric code: rigidity and sharp gadget}\label{sec:sm-toric}

Write $d=2\rho+1$.  A minimum noncontractible dual logical loop has weight
$d$.

\begin{smlemma}[Half-distance rigidity]
\label{lem:sm-toric-rigidity}
If two weight-$(\rho+1)$ errors have the same syndrome and different logical
classes, each of the two logical cosets has exactly one
weight-$(\rho+1)$ representative.
\end{smlemma}

\begin{proof}
\emph{Step 1: classify every logical product below weight $d+2$.}
Let $E$ and $F$ be the two errors.  Their product
$L=EF$ has zero syndrome, a nontrivial logical label, and
$\operatorname{wt}(L)\leq2\rho+2=d+1$. For either
Calderbank--Shor--Steane (CSS) component of $L$, choose the $d$
parallel cuts transverse to its nontrivial winding.  The component must
cross every cut, hence has at least $d$ edges.  Equality forces one
crossing per cut.  A transverse step changes lattice row (or column) and
must be undone before the loop closes, adding at least two edges.
Therefore every nontrivial single-component chain of weight at most
$d+1$ is a straight geodesic loop of weight $d$.  If both CSS components
were nontrivial, their physical geodesics would be a primal and a dual
representative.  Parallel representatives are edge-disjoint and
transverse representatives share at most one physical edge, so their
Pauli union has weight at least $d+d-1=2d-1>d+1$.  A nonzero second,
contractible CSS component must contain at least two physical edges
outside the straight geodesic: even incidence forbids a cycle supported
only on its disconnected transverse edges.  Its Pauli union with the
geodesic therefore has weight at least $d+2$.  Thus $L$ itself is one
straight primal or dual loop $C$.

\emph{Step 2: classify the two endpoints.}
At a position of $C$ occupied by both $E$ and $F$, their Pauli labels
must differ, so the position remains in $EF$ and contributes one excess
unit to $\operatorname{wt}(E)+\operatorname{wt}(F)-\operatorname{wt}(EF)$.
Let $q$ count these positions.  Outside $C$,
$EF$ is the identity, so a common occupied position has equal Pauli
labels and contributes two excess units; let their number be $u$.
There can be no position occupied by only one endpoint outside $C$.
Consequently
\begin{equation}
 2(\rho+1)=\operatorname{wt}(E)+\operatorname{wt}(F)=d+q+2u.
\end{equation}
Hence $q+2u=1$, so $q=1$ and $u=0$.  Both endpoints are supported on
$C$, they overlap at exactly one site, and their remaining supports
partition $C$.

\emph{Step 3: exclude additional minimum representatives.}
Fix $F$ and suppose $E'$ is another weight-$(\rho+1)$ representative in
the logical class of $E$.  The pair $(E',F)$ again satisfies Steps 1 and
2, so $E'F$ is a straight loop $C'$, and $F$ is supported on both $C$
and $C'$.  Distinct parallel straight loops are disjoint, while
perpendicular ones intersect in only one edge; neither can contain the
$\rho+1\geq2$ support of $F$.  Thus $C'=C$, and then
$E'=FC=E$ as phase-free Pauli operators.  Interchanging $E$ and $F$
gives the same conclusion for the other class.  Therefore each of the
two logical cosets has exactly one weight-$(\rho+1)$ representative.
\end{proof}

Every eligible half-distance class is therefore a singleton.  By
Theorem~\ref{thm:sm-spectrum}, if its minimum monomial is unique then MP
and eventual MLD share that class; if several monomials tie, the MP
winner set contains every lexicographic MLD resolution.  Thus no
higher-weight coefficient can create a strict half-distance split, and
\begin{equation}
 m_d^{\rm toric}\geq \rho+2=\frac{d+3}{2}.
\end{equation}

To attain this order for $d\geq5$, choose a straight loop
$L=\prod_{a=0}^{d-1}X_{v_{a,0}}$, let
\begin{equation}
 Q=\{0,2,3\},\qquad A=\{4,5,\ldots,\rho+2\},
\end{equation}
and set
\begin{equation}
 E_A=\prod_{a\in A}X_{v_{a,0}}
     \prod_{a\in Q}Z_{v_{a,0}},\qquad E_B=E_AL,
\label{eq:sm-toric-gadget}
\end{equation}
where $v_{a,b}$ is the vertical edge from $(a,b)$ to $(a,b+1)$, with
coordinates understood modulo $d$.

\begin{figure}[h]
\centering
\begin{minipage}{0.82\textwidth}
\centering
\includegraphics[width=\linewidth]{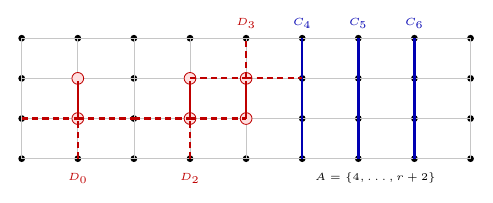}
\end{minipage}
\caption{Local coordinates for the toric union calibration.  Blue
vertical columns are the mutually disjoint parity detectors $C_a$ for
the short dual chain.  Solid red edges form $Q_v$; dashed red segments
trace the three vertex stars $D_0,D_2,D_3$ charged by its boundary.
Their physical edge sets are pairwise disjoint and avoid all $C_a$,
$a\in A$.  Periodic
identification and unshown middle columns do not alter these local
incidences.}
\label{fig:sm-toric-calibration}
\end{figure}

\begin{smlemma}[Exact gadget spectra]
\label{lem:sm-toric-spectra}
The syndrome of Eq.~(\ref{eq:sm-toric-gadget}) has exactly two logical
classes at weight $\rho+2$, with
\begin{equation}
 P_{A,d}=2x^{\rho-1}z^3,\qquad
 P_{B,d}=x^{\rho-1}y^3.
 \label{eq:sm-toric-spectra}
\end{equation}
\end{smlemma}

\begin{proof}
Separate a Pauli error into its dual $X/Y$ chain $R_X$ and primal $Z/Y$
chain $R_Z$; their intersection is the $Y$ support, and therefore
\begin{equation}
 \operatorname{wt}(R_X,R_Z)=|R_X\cup R_Z|
 =|R_X|+|R_Z|-|R_X\cap R_Z|.
 \label{eq:sm-toric-union}
\end{equation}
With $\chi^X_e,\chi^Z_e\in\mathbb F_2$ denoting their edge bits, the complete
local constraints are
\begin{equation}
 \bigoplus_{e\in\partial p}\chi^X_e=s_p^Z
 \quad\text{for every plaquette }p,\qquad
 \bigoplus_{e\ni q}\chi^Z_e=s_q^X
 \quad\text{for every vertex }q.
\label{eq:sm-toric-local-parity}
\end{equation}
The two remaining parity equations are the primal and dual homology
bits.  Every equality classification below is obtained by restricting
these displayed equations to the charged support.

\emph{Step 1: dual-chain equality cases.}
The syndrome fixes the boundary of $R_X$ to the two endpoints of the
interval $A$ on the vertical loop.  Intersecting with the $d$ transverse
cuts gives lower bounds $\rho-1$ and $d-(\rho-1)=\rho+2$ in the two homology
sectors.  Equality crosses each required cut exactly once and crosses no
other cut.  Hence the unique equality chains are $A$ and
$L\setminus A$; a transverse excursion adds two edges.

\emph{Step 2: primal T-joins.}
The set $Q=\{0,2,3\}$ fixes six vertex terminals, paired locally in
three separated unit intervals.  Crossing the three cuts isolating these
intervals gives $|R_Z|\geq3$.  Equality requires one edge in each
interval and none elsewhere.  There are exactly two compatible choices:
the three edges $\{v_{a,0}:a\in Q\}$ and the chain obtained from them by
the single adjacent plaquette boundary.  These are in the same required
primal homology sector.  Any other pairing crosses one of the isolating
cuts at least twice or leaves the strip and must return, so it has length
at least five.

\emph{Step 3: disjoint-cut calibration of the Pauli union.}
The Pauli union in Eq.~(\ref{eq:sm-toric-union}) requires a joint lower
bound, because a longer pair could create extra $Y$ overlap.  We therefore
lower-bound the union itself.

For each column $a$, let
\begin{equation}
 C_a=\{v_{a,b}:b\in\mathbb Z_d\}.
\end{equation}
As a physical edge set, $C_a$ is a primal noncontractible cycle and
therefore a mod-two intersection detector for dual chains.
Comparing $R_X$ with the reference chain $A$ shows, modulo two,
\begin{equation}
 |R_X\cap C_a|=
 \begin{cases}
 1,&a\in A,\\0,&a\notin A
 \end{cases}
\quad (R_X\sim A),
\label{eq:sm-toric-column-parity}
\end{equation}
with the two values interchanged in the $L\setminus A$ sector.  Indeed,
the symmetric difference with the reference chain is a closed dual
chain of trivial homology and has even intersection with every $C_a$.

For the primal chain define three vertex cuts
\begin{equation}
 D_0=\delta\{(0,0)\},\qquad
 D_2=\delta\{(2,0)\},\qquad
 D_3=\delta\{(3,1)\}.
\label{eq:sm-toric-vertex-cuts}
\end{equation}
Each selected vertex is a terminal of $\partial R_Z$, so the handshaking
identity gives
\begin{equation}
 |R_Z\cap D_i|=1\pmod2,\qquad i=0,2,3.
\label{eq:sm-toric-vertex-parity}
\end{equation}
The three vertices are pairwise nonadjacent, so the $D_i$ are
edge-disjoint.  They use vertical columns $0,2,3$ and horizontal edges
only; consequently they are also disjoint from every $C_a$ with
$a\in A=\{4,\ldots,\rho+2\}$.  Eqs.~(\ref{eq:sm-toric-column-parity})
and~(\ref{eq:sm-toric-vertex-parity}) therefore charge $\rho-1+3=\rho+2$
pairwise disjoint physical edge sets.  Since a $Y$ edge is still one
physical edge, it can discharge only one of these cuts.  Thus
\begin{equation}
 |R_X\cup R_Z|\geq \rho+2
 \qquad(R_X\sim A).
\label{eq:sm-toric-calibration}
\end{equation}

Equality in Eq.~(\ref{eq:sm-toric-calibration}) means that the union
contains exactly one edge in every charged cut and no edge outside
them.  We now solve the dual parity equations on this support.  For
$d\geq7$, the right terminal plaquette is incident to only
$v_{\rho+2,0}$ among the charged edges, so its odd equation forces that
edge.  The next plaquette has only $v_{\rho+2,0}$ and $v_{\rho+1,0}$
available; its even equation forces the latter.  Iterating right to left
gives
\begin{equation}
 \chi^X_{v_{a,b}}=\delta_{b,0},\qquad a=4,\ldots,\rho+2.
\label{eq:sm-toric-column-peel}
\end{equation}
These edges already have the prescribed dual boundary.  The remaining
$R_X$ support, if any, must be a cycle in
$D_0\cup D_2\cup D_3$.  In the dual graph this union is one isolated
four-cycle and two four-cycles wedged at one vertex.  Its cycle space
has basis $D_0,D_2,D_3$; every nonzero cycle therefore uses all four
edges of at least one $D_i$.  Equality permits at most one edge from
each $D_i$, so the remaining cycle is empty and $R_X=A$.

For $d=5$, periodic closure makes $D_0$ incident to the right terminal.
Direct substitution in the same local plaquette equations gives the
complete base table (five choices in $C_4$ and, independently, no edge
or one of four edges in each $D_i$)
\begin{equation}
\begin{array}{c|ccccc}
\text{edge chosen from }C_4
 &v_{4,0}&v_{4,1}&v_{4,2}&v_{4,3}&v_{4,4}\\ \hline
\#\{R_X:\partial R_X=\partial A,\ |R_X\cap D_i|\leq1\}
 &1&0&0&0&0
\end{array},
\label{eq:sm-toric-d5-dual-table}
\end{equation}
whose unique entry is $R_X=\{v_{4,0}\}=A$.  This closes the periodic
base case without an all-code enumeration.

It remains to solve the primal boundary equations on the charged
support.  At every $a\in A$, an endpoint of the isolated edge
$v_{a,0}$ has even prescribed parity and no other available incident
edge, so $R_Z$ contains no $C_a$ edge.  Every remaining choice is
exhausted by
\begin{equation}
\begin{array}{c|cccc}
D_0\text{ edge}&h_{d-1,0}&h_{0,0}&v_{0,d-1}&v_{0,0}\\ \hline
\text{correct terminal pair}&0&0&0&1
\end{array}
\qquad
\begin{array}{c|cccc}
 &h_{2,1}&h_{3,1}&v_{3,0}&v_{3,1}\\ \hline
h_{1,0}&0&0&0&0\\
h_{2,0}&1&0&0&0\\
v_{2,d-1}&0&0&0&0\\
v_{2,0}&0&0&1&0
\end{array}.
\label{eq:sm-toric-primal-table}
\end{equation}
An entry 1 means that the selected edges have precisely the required
terminal set; 0 means that at least one displayed vertex equation
fails.  The two surviving perfect $T$-joins are
\begin{equation}
 Q_v=\{v_{0,0},v_{2,0},v_{3,0}\},\qquad
 Q_v'=\{v_{0,0},h_{2,0},h_{2,1}\}.
\label{eq:sm-toric-local-solutions}
\end{equation}
Their difference is the unit-square boundary.  Table
(\ref{eq:sm-toric-primal-table}) checks every allowed edge in the three
charged stars, so no further solution is implicit.  Hence the $A$
sector has exactly two equality cases, both of composition
$(\rho-1,0,3)$.

In the complementary sector,
Eq.~(\ref{eq:sm-toric-column-parity}) charges the $\rho+2$ mutually
disjoint columns $C_a$ with $a\notin A$.  Equality allows no horizontal
edge, so the dual boundary equations force
$R_X=L\setminus A$.  The primal chain must be supported inside it; its
boundary equations uniquely give $R_Z=Q_v$.  This yields one error of
composition $(\rho-1,3,0)$, whereas $Q_v'$ has two horizontal edges and
union weight $\rho+4$.  Changing either remaining toric homology bit
requires a noncontractible component of at least $d>\rho+2$ edges.
The equality classification is therefore the finite table
\begin{equation}
\begin{array}{c|c|c|c|c}
\text{sector}&R_X&R_Z&(n_X,n_Y,n_Z)&\text{multiplicity}\\ \hline
 A&A&Q_v&(\rho-1,0,3)&1\\
 A&A&Q_v'&(\rho-1,0,3)&1\\
 B&L\setminus A&Q_v&(\rho-1,3,0)&1
\end{array}
\label{eq:sm-toric-exhaustive-table}
\end{equation}
All other logical sectors have minimum weight greater than $\rho+2$.
This proves Eq.~(\ref{eq:sm-toric-spectra}) directly in physical-union
units and exhausts every equality case.
\end{proof}

If $z<y$, MP uniquely chooses $B$, whereas $y^3<2z^3$ makes MLD uniquely
choose $A$.  Thus the lower bound is sharp.  At $d=3$, direct equality
enumeration is also finite enough to display.  Order qubits as
$h_{0,0},h_{1,0},\ldots,h_{2,2},v_{0,0},\ldots,v_{2,2}$.
Row reduction of Eq.~(\ref{eq:sm-toric-local-parity}) and the four
logical equations at union weight three gives
\begin{equation}
\begin{array}{c|c|c|c}
\text{class}&\text{representative}&(n_X,n_Y,n_Z)&\text{count}\\ \hline
0100&X_{h_{0,0}}X_{h_{1,0}}Z_{h_{0,1}}&(2,0,1)&1\\
0100&Z_{h_{0,1}}X_{v_{1,0}}X_{v_{1,2}}&(2,0,1)&1\\
1100&X_{h_{1,0}}Y_{h_{0,1}}X_{h_{0,2}}&(2,1,0)&1\\
\text{other 14 classes}&\text{none}&-&0
\end{array}
\label{eq:sm-toric-d3-table}
\end{equation}
Thus the two spectra are $2x^2z$ and $x^2y$, whose strict region contains
the main article's toric region. Hence
\begin{equation}
 m_d^{\rm toric}=\frac{d+3}{2}
\end{equation}
for every odd $d$ in that region.

\section{Separable HGP/BB family: exact two-order delay}\label{sec:sm-hgp}

\subsection{All-size rigidity for a separable BB/HGP family}
\label{sec:sm-bb-hgp-all-q}

We now prove an all-size result for a hypergraph-product family
\cite{TillichZemor2014}. Put $n=3q$, let $\mathbb F_2$ denote the binary field, and let
\begin{equation}
 A_q=I+S+S^2:\mathbb F_2^n\longrightarrow\mathbb F_2^n,
 \qquad q\geq1,
 \label{eq:sm-hgp-seed}
\end{equation}
where $S$ is the cyclic shift.  The product complex has
\begin{align}
 Q_2&=V_1\otimes V_1,\nonumber\\
 Q_1&=(V_1\otimes V_0)\oplus(V_0\otimes V_1),\nonumber\\
 Q_0&=V_0\otimes V_0,
 \label{eq:sm-hgp-complex}
\end{align}
with $V_1=V_0=\mathbb F_2^n$ and
\begin{align}
 \partial_1(u,v)&=(A_q\otimes I)u+(I\otimes A_q)v,\nonumber\\
 \partial_2(w)&=((I\otimes A_q)w,(A_q\otimes I)w).
 \label{eq:sm-hgp-boundaries}
\end{align}
Thus $\partial_1\partial_2=0$. Taking
$H_X=\partial_1$ and $H_Z=\partial_2^T$ gives, up to a permutation of
the two qubit blocks and tensor factors, the bivariate-bicycle (BB) code
\begin{equation}
 \mathcal H_q=Q(1+y+y^2,1+x+x^2;3q,3q).
 \label{eq:sm-hgp-family}
\end{equation}
where $Q(a,b;\ell,m)$ denotes the standard polynomial presentation of a BB
code with check polynomials $a,b$ on an $\ell\times m$ periodic cell.

\begin{smlemma}[Seed-code equality structure]
\label{lem:sm-hgp-seed}
The codes $\ker A_q$ and $\ker A_q^T$ both have dimension two and distance
$2q$.  Each has exactly three nonzero words, all of weight $2q$.
Moreover, no coordinate vector belongs to $\operatorname{im}A_q$ or
$\operatorname{im}A_q^T$.
\end{smlemma}

\begin{proof}
Writing $c=(c_i)_{i\in\mathbb Z_n}$, the equation $hc=0$ is, up to a
cyclic reversal of the indices,
\begin{equation}
 c_i+c_{i-1}+c_{i-2}=0.
 \label{eq:sm-hgp-recurrence}
\end{equation}
It implies $c_i=c_{i-3}$.  Hence a kernel word is fixed by its first two
bits.  The three nonzero choices repeat one of the patterns $101$,
$011$, and $110$, and therefore have weight $2q$.  Reversing the cyclic
orientation proves the same statement for $A_q^T$.

For every coordinate $j$, one of the three nonzero words in $\ker A_q^T$
has a one at $j$.  Its inner product with $e_j$ is one, whereas it is
orthogonal to $\operatorname{im}A_q$. Thus
$e_j\notin\operatorname{im}A_q$. Interchanging $A_q$ and $A_q^T$ gives the
last assertion.
\end{proof}

Let $K=\ker A_q$ and
$C=\operatorname{coker}A_q:=V_0/\operatorname{im}A_q$.
The Kunneth map for
Eq.~\eqref{eq:sm-hgp-complex} is
\begin{equation}
 \begin{aligned}
 \Phi:H_1(Q)&\longrightarrow(K\otimes C)\oplus(C\otimes K),\\
 [(u,v)]&\longmapsto((I\otimes\pi)u,(\pi\otimes I)v),
 \end{aligned}
 \label{eq:sm-hgp-kunneth}
\end{equation}
where $\pi:V_0\to C$ is the quotient map.  For completeness, this map
is an isomorphism: choose splittings
$V_1=K\oplus L$ and
$V_0=\operatorname{im}A_q\oplus C_0$. The restriction
$A_q|_L:L\to\operatorname{im}A_q$ is an isomorphism, so all tensor-product
summands containing the pair $L\xrightarrow{A_q}\operatorname{im}A_q$
are contractible.  The two surviving degree-one summands are precisely
$K\otimes C_0$ and $C_0\otimes K$, which are identified with the target
of Eq.~\eqref{eq:sm-hgp-kunneth}.

\begin{smlemma}[Complete minimum-cycle classification]
\label{lem:sm-hgp-min-cycles}
Here a Pauli centralizer is an operator commuting with every stabilizer.
Every nontrivial $Z$-type centralizer of $\mathcal H_q$ has weight at
least $2q$.  Equality holds only for
\begin{equation}
 (c\otimes e_j,0)
 \quad\text{or}\quad
 (0,e_i\otimes c),
 \label{eq:sm-hgp-min-z}
\end{equation}
where $c$ is one of the three nonzero words of $\ker A_q$ and
$i,j\in\mathbb Z_n$.  Dually, every minimum $X$-type centralizer is
\begin{equation}
 (e_i\otimes c^\vee,0)
 \quad\text{or}\quad
 (0,c^\vee\otimes e_j),
 \label{eq:sm-hgp-min-x}
\end{equation}
where $c^\vee$ is one of the three nonzero words of $\ker A_q^T$.
Consequently there are exactly $18q$ minimum representatives of each
Calderbank--Shor--Steane (CSS) type.
\end{smlemma}

\begin{proof}
Let $(u,v)\in\ker\partial_1$ represent nonzero homology.  At least one
component of $\Phi[(u,v)]$ is nonzero.  Suppose first that
$(I\otimes\pi)u\ne0$.  A linear functional on $C$ can then be chosen so
that contraction in the second tensor factor gives a nonzero word
\begin{equation}
 c=(I\otimes\lambda)u\in\ker A_q.
\end{equation}
Every nonzero coordinate of $c$ requires a nonempty first-factor slice
of $u$.  Therefore
\begin{equation}
 2q\leq |c|\leq |u|\leq |u|+|v|.
 \label{eq:sm-hgp-slice-bound}
\end{equation}
The other Kunneth component gives the same bound with $u$ and $v$
interchanged.  This proves the lower bound.

If $|u|+|v|=2q$ in the first case, every inequality in
Eq.~\eqref{eq:sm-hgp-slice-bound} is an equality, so $v=0$.  The cycle
condition becomes $(A_q\otimes I)u=0$. Expanding
$u=\sum_j u_j\otimes e_j$ shows that every nonzero $u_j$ lies in
$\ker A_q$ and hence has weight at least $2q$. Since $|u|=2q$, exactly
one $u_j$ is nonzero and it is a minimum seed word.  This gives the
first form in Eq.~\eqref{eq:sm-hgp-min-z}; the second Kunneth component
gives the other form.  Lemma~\ref{lem:sm-hgp-seed} ensures that every
displayed coordinate vector is nonzero in the relevant cokernel, so
none of these operators is a boundary.

Applying the identical argument to the dual complex gives
Eq.~\eqref{eq:sm-hgp-min-x}.  There are three seed words, $n$ coordinate
choices, and two orientations, hence $6n=18q$ representatives of each
CSS type.
\end{proof}

\begin{smlemma}[Pauli purity at minimum distance]
\label{lem:sm-hgp-pauli-purity}
Every weight-$2q$ nontrivial Pauli centralizer of $\mathcal H_q$ is
purely $X$ or purely $Z$.  Hence the complete all-Pauli minimum spectrum
contains exactly $36q$ representatives and contains no $Y$ support.
\end{smlemma}

\begin{proof}
Write a Pauli modulo phase as $X^aZ^b$.  Centrality is equivalent to
$a\in\ker\partial_2^T$ and $b\in\ker\partial_1$.  If the Pauli has a
nontrivial logical label, at least one of $[a]$ and $[b]$ is nonzero.
Suppose $[a]\ne0$.  Lemma~\ref{lem:sm-hgp-min-cycles} and its dual give
$|a|\geq2q$.  If the Pauli weight is $2q$, then $a$ is one of the lines
in Eq.~\eqref{eq:sm-hgp-min-x} and
$\operatorname{supp}b\subseteq\operatorname{supp}a$.

For example, if $a=(e_i\otimes c^\vee,0)$, then
$b=(e_i\otimes v,0)$ for some $v$.  Its cycle equation reads
\begin{equation}
 (A_q\otimes I)(e_i\otimes v)=A_qe_i\otimes v=0.
\end{equation}
Every column $he_i$ is nonzero, so $v=0$.  The other orientation is
identical.  Thus $b=0$.  If instead $[b]\ne0$, the dual argument uses
the nonzero columns of $A_q^T$ and gives $a=0$. The count now follows
from Lemma~\ref{lem:sm-hgp-min-cycles}.
\end{proof}

The preceding lemmas also prove directly that
$\mathcal H_q$ has parameters $[[18q^2,8,2q]]$: there are $2n^2$
qubits, Eq.~\eqref{eq:sm-hgp-kunneth} and its dual give
$\dim H_1=2(\dim K)(\dim C)=8$, and the displayed operators attain the
distance lower bound.

\begin{smtheorem}[All-$q$ HGP half-distance rigidity]
\label{thm:sm-hgp-all-q}
For every $q\geq1$ and every interior memoryless Pauli direction,
the strict MP--MLD disagreement onset of $\mathcal H_q$ obeys
\begin{equation}
 m(\mathcal H_q)\geq q+1=\frac{d_q}{2}+1.
 \label{eq:sm-hgp-delay}
\end{equation}
Equivalently, no strict decision split is possible at or below half
distance.
\end{smtheorem}

\begin{proof}
The universal distance obstruction excludes weights below $q$.  It
remains to consider a syndrome with representatives of weight at most
$q$ in more than one logical class.  If $E$ and $F$ are representatives
in two such classes, then $EF$ is a nontrivial centralizer and
\begin{equation}
 2q=d_q\leq |EF|\leq |E|+|F|\leq2q.
 \label{eq:sm-hgp-balanced-equality}
\end{equation}
All inequalities are equalities.  In particular $|E|=|F|=q$, their
supports are disjoint, and Lemma~\ref{lem:sm-hgp-pauli-purity} makes
$EF$ a pure minimum logical.  Thus $E$ and $F$ are complementary
$q$-subsets of one pure logical line and have the same Pauli
composition: either $(q,0,0)$ or $(0,0,q)$.

Fixing one representative in each of two classes and repeating the
same argument with every other weight-$q$ representative in the fiber
shows that all of them have the same pure composition.  Hence every
represented logical class has the same individual peak monomial,
$x^q$ or $z^q$.  MP therefore ties all represented classes.  Degenerate
MLD may select a class with larger multiplicity, but every such class
remains in the MP-optimal set.  The two winner sets cannot be disjoint,
independently of tie breaking and of higher-weight coefficients.  The
half-distance layer consequently has no strict split, proving
Eq.~\eqref{eq:sm-hgp-delay}.
\end{proof}

Exhaustive enumeration for $q=1,2,3$ confirms
$36q$ minimum Pauli centralizers and
$36q\binom{2q}{q}$ ordered half factorizations and checks every
half-distance fiber at those sizes.


\subsection{The excess-two layer of the separable HGP/BB family}
\label{sec:sm-hgp-excess-two}

Let $\mathcal H_q=[[18q^2,8,2q]]$ be the family defined in
Sec.~\ref{sec:sm-bb-hgp-all-q}, and write $d=2q$.  A competition between
two different logical classes at physical weight $q+1$ is generated by a
nontrivial centralizer $U$ of weight at most $2q+2$.

We first sharpen the seed-code facts used below.  For
$a\in\mathbb F_2^{3q}$, let $p_\alpha(a)$ be the parity of its coordinates in
the residue class $\alpha$ modulo three.  Orthogonality to the two-dimensional
kernel of $A_q^T$ gives
\begin{equation}
 a\in\operatorname{im}A_q
 \quad\Longleftrightarrow\quad
 p_0(a)=p_1(a)=p_2(a).
 \label{eq:sm-hgp-image-parities}
\end{equation}
The reverse implication follows because both sides have dimension $3q-2$.
In particular, no singleton lies in $\operatorname{im}A_q$, and a weight-two
vector lies in $\operatorname{im}A_q$ precisely when its two positions have
the same residue modulo three.  For $q\ge2$, every column $he_j$ has three
distinct entries and two different columns of $A_q$ have different supports.

\begin{smlemma}[Four-step CSS systolic gap]
\label{lem:sm-hgp-css-four-step-gap}
For $q\ge2$, every nontrivial $Z$-type centralizer of $\mathcal H_q$ with
weight at most $2q+2$ is one of the weight-$2q$ lines in
Eq.~\eqref{eq:sm-hgp-min-z}.  The dual statement holds for $X$-type
centralizers.
\end{smlemma}

\begin{proof}
Let $(u,v)\in\ker\partial_1$ have nonzero first Kunneth component.  The
slice bound in Eq.~\eqref{eq:sm-hgp-slice-bound} gives $|u|\ge2q$; hence a
cycle of total weight at most $2q+2$ has $|v|\le2$.  Expand both blocks in
second-factor columns. If $v=0$, every column of $u$ belongs to $\ker A_q$.
A nonzero column has weight $2q$, and two such columns already have weight
$4q>2q+2$.  Thus there is exactly one nonzero column and the cycle is a
minimum line.

A weight-one $v$ is impossible: $(I\otimes A_q)v$ has three columns whose
first-factor slices are singletons, whereas every column of
$(A_q\otimes I)u$ lies in $\operatorname{im}A_q$, contradicting
Eq.~\eqref{eq:sm-hgp-image-parities}.  Now let $|v|=2$.  If the two bits of
$v$ have different second coordinates, the two distinct three-column
supports produced by $I\otimes A_q$ have a column occupied by only one bit,
giving the same singleton contradiction.  The bits must therefore have a
common second coordinate and distinct first coordinates $a,a'$ in the same
residue class.  Each of the three affected columns of $u$ then solves
\begin{equation}
 A_q u_j=e_a+e_{a'}.
 \label{eq:sm-hgp-affine-column}
\end{equation}
Let $\mathcal A$ be this four-element affine solution set, and let
$w_1\le w_2\le w_3\le w_4$ be its ordered weights.  The syndrome on the
right of Eq.~\eqref{eq:sm-hgp-affine-column} has weight two, so $w_1\ge2$.
Any two distinct elements of $\mathcal A$ differ by a nonzero word of
$\ker A_q$ and consequently
\begin{equation}
 w_1+w_2\ge2q.
 \label{eq:sm-hgp-affine-pair}
\end{equation}

If the three affected columns of $u$ are equal and every unaffected column
vanishes, then $u=a\otimes he_j$ and $v=ha\otimes e_j$ for some
$a\in\mathcal A$; the pair is exactly the boundary
$\partial_2(a\otimes e_j)$.  If an unaffected column is nonzero, it is a
kernel word of weight at least $2q$, while each of the three affected
columns has weight at least two.  Including $|v|=2$ gives total weight at
least $2q+8$, a contradiction.  We may therefore assume that all unaffected
columns vanish.  If the three affected columns are not all equal, they have
total
weight at least
\begin{equation}
 2w_1+w_2\ge2q+w_1\ge2q+2.
\end{equation}
Adding the two bits of $v$ gives total weight at least $2q+4$, a
contradiction.  This proves the claim when the first Kunneth component is
nonzero.  Interchanging the tensor factors treats the second component,
and dualizing proves the $X$ statement.
\end{proof}

\begin{smlemma}[Pauli four-step gap]
\label{lem:sm-hgp-pauli-four-step-gap}
For $q\ge2$, every nontrivial all-Pauli centralizer of $\mathcal H_q$ with
weight at most $2q+2$ is a pure-$X$ or pure-$Z$ minimum logical line.
Equivalently, there are no nontrivial centralizer representatives of weight
$2q+1$ or $2q+2$.
\end{smlemma}

\begin{proof}
Write the phase-free Pauli as $P=X^aZ^b$.  Because the code is CSS,
commutation with the two stabilizer types implies separately that $a$ is an
$X$ cycle and $b$ is a $Z$ cycle.  Their logical labels form a direct sum,
so at least one is nonzero; suppose it is the label of $a$.  Since
$|a|\leq|\operatorname{supp}(a)\cup\operatorname{supp}(b)|=|P|\leq2q+2$,
Lemma~\ref{lem:sm-hgp-css-four-step-gap} makes $a$ a minimum $X$ line
$\Lambda$ of weight $2q$.  Consequently $b$ has at most two occupied sites
outside $\Lambda$.

We now prove that any $Z$ cycle with this property is zero.  Take first
$\Lambda=(e_i\otimes c^\vee,0)$ and write the $Z$ cycle as $(u,v)$ in the
two HGP qubit blocks.  The whole $v$ block lies outside $\Lambda$, hence
$|v|\leq2$.  If $v=0$, each nonzero second-factor column of $u$ belongs to
$\ker A_q$, has weight $2q$, and intersects $\Lambda$ in at most one site.
It would therefore have at least $2q-1\geq3$ sites outside $\Lambda$, which
is impossible; thus $u=0$.

If $|v|=1$, one of the three column slices of $(I\otimes A_q)v$ is a
singleton. Every slice of $(A_q\otimes I)u$ lies in
$\operatorname{im}A_q$, whereas Eq.~\eqref{eq:sm-hgp-image-parities}
excludes a singleton from $\operatorname{im}A_q$.
Hence this case is impossible.  Suppose finally that $|v|=2$.  If its two
sites have different second coordinates, the two distinct three-term
column supports of $A_q$ leave at least one singleton slice, giving the same
contradiction.  The sites must therefore have one common second coordinate,
and their first coordinates $a,a'$ have the same residue modulo three.
Exactly three columns are then affected, and each obeys
$A_q u_j=e_a+e_{a'}$. The two sites of $v$ already exhaust the outside-site
budget, so every $u_j$ is supported on $\Lambda$ and is either $0$ or the
singleton $e_i$.  Its image has weight zero or three, never two, contradicting
the displayed equation.  Thus $b=0$.

Interchanging the tensor factors treats the other orientation of $\Lambda$,
and exchanging $X$ and $Z$ treats the case in which $b$ has nonzero logical
label.  Therefore every nontrivial Pauli centralizer through weight $2q+2$
is one of the pure minimum logical lines.
\end{proof}

\begin{smtheorem}[Excess-two balance]
\label{lem:sm-hgp-excess-two-balance}
For every $q\ge2$, no strict MP--MLD split occurs at physical weight $q+1$,
for any interior memoryless Pauli direction.  Consequently
\begin{equation}
 m(\mathcal H_q)\ge q+2.
 \label{eq:sm-hgp-conditional-lower}
\end{equation}
\end{smtheorem}

\begin{proof}
Fix a syndrome whose minimum-weight fiber $\mathcal V_s$ occurs at weight
$q+1$ and contains at least two logical classes.  Define its competition
graph $G_s$ to have vertex set $\mathcal V_s$, joining two vertices exactly
when their logical labels differ.  Thus $G_s$ is complete multipartite, with
its parts equal to the logical classes represented in the fiber; in
particular, $G_s$ is connected.

For any edge $E--F$, the product $U=E^\dagger F$ is a nontrivial
centralizer.  Lemma~\ref{lem:sm-hgp-pauli-four-step-gap} makes $U$ a pure-$X$
or pure-$Z$ minimum logical line of weight $2q$.  The joint excess is
\begin{equation}
 |E|+|F|-|U|=2.
 \label{eq:sm-hgp-excess-two}
\end{equation}
On $\operatorname{supp}(U)$, let $c$ count sites on which both factors are nonidentity;
outside $\operatorname{supp}(U)$, the identity product forces the factors to agree, and
let $o$ count their common nonidentity sites.  Summing the four possible
single-site occupancy patterns gives the exact phase-free Pauli identity
\begin{equation}
 c+2o=2,
 \label{eq:sm-hgp-cr-two}
\end{equation}
so there are only two cases.

If $(c,o)=(0,1)$, each factor occupies one common outside site and $q$ sites
of $\operatorname{supp}(U)$.  On the latter sites exactly one factor is occupied and must
carry the Pauli type of the pure line $U$; on the former site both factors
carry the same Pauli.  Hence $E$ and $F$ have identical Pauli compositions.

If $(c,o)=(2,0)$ and, without loss of generality, $U$ is pure $X$, each
factor is supported inside $U$: it has $q-1$ sites carrying $X$ and two
overlapping twisted sites carrying $Y$ or $Z$, with the two local Paulis
multiplying to $X$.  Distinct
minimum pure-$X$ lines in $\mathcal H_q$ intersect in at most $q$ sites.
Indeed, parallel lines at different transverse coordinates are disjoint,
while two distinct period-three seed words on the same line share exactly
one of the three residue classes, of size $q$; lines in the two HGP qubit
blocks are disjoint.  Hence a weight-$(q+1)$ vertex cannot be contained in
two distinct minimum pure-$X$ lines.  For a fixed line, multiplication by
that line uniquely determines the opposite endpoint.  The dual statement
holds for pure $Z$ lines.

Such a vertex cannot also participate in a $(0,1)$ edge of the same CSS
type.  Relative to a pure-$X$ product, a $(2,0)$ vertex has $Z$ support on
both twisted sites, whereas a $(0,1)$ vertex is pure $X$ on the logical
line and has $Z$ support on at most its one common outside site.  These
support descriptions are incompatible.  Exchanging $X$ and $Z$ gives the
dual statement.

Nor can a vertex be incident to both a pure-$X$ and a pure-$Z$ product
edge.  If $E^\dagger F=U_X$ and $E^\dagger G=U_Z$, then
$F^\dagger G=U_XU_Z$ has nonzero components in both CSS logical summands
and weight at most $|F|+|G|=2q+2$.  This contradicts
Lemma~\ref{lem:sm-hgp-pauli-four-step-gap}.  Therefore every vertex incident
to a $(2,0)$ edge has degree one.  Since $G_s$ is complete multipartite and
connected, the presence of such an edge forces $G_s$ to consist of precisely
that isolated pair.  Its two represented classes are classwise singletons,
so their leading coset masses equal their peak monomials and MP and MLD have
the same winner set whenever the monomials are unequal.  On their equality
locus, both classes remain in the MP tie set, so higher-order coset terms can
select at most one of those same classes and still cannot create strict
best-tie separation.

If no $(2,0)$ edge occurs, every edge is of type $(0,1)$ and hence joins
vertices with identical composition.  Connectivity of $G_s$ then makes the
entire fiber composition-pure.  Every represented class has the same peak
monomial, so all are MP-optimal, whereas MLD selects those with largest
multiplicity; the two optimal sets necessarily intersect.  Because the
single-error monomials are exactly equal, not merely equal in exponent,
higher-order MLD terms can only refine a choice inside this persistent MP tie
set.  Thus neither possible graph structure gives strict separation at
weight $q+1$.  Together with the universal half-distance obstruction, this proves
Eq.~\eqref{eq:sm-hgp-conditional-lower}.
\end{proof}

For $q=1$, the separate complete certificate gives
$m(\mathcal H_1)=3=q+2$.  Combining it with
Theorem~\ref{lem:sm-hgp-excess-two-balance} therefore proves the all-size
lower bound
\begin{equation}
 m(\mathcal H_q)\ge q+2\quad\text{for every }q\ge1.
 \label{eq:sm-hgp-all-q-strengthened-lower}
\end{equation}
For $q=1$, the separate complete certificate supplies the base case of the
onset bound; the four-step Pauli gap applies to $q\ge2$.

\subsection{A uniform detector-saturating onset gadget}
\label{sec:sm-hgp-uniform-gadget}

We now match Eq.~\eqref{eq:sm-hgp-all-q-strengthened-lower}.  Put
$n=3q$ and label each HGP qubit by a block $L$ or $R$ and a coordinate
$(i,j)\in\mathbb Z_n^2$.  For $q\ge2$, define
\begin{equation}
 r_k=4+3k,\qquad k=0,\ldots,q-2.
 \label{eq:sm-hgp-rk}
\end{equation}
The reference error $M_R$ has $X$ on
\begin{equation}
 R(1,n-1),\ R(2,n-1),\ R(3,n-1)
 \label{eq:sm-hgp-mr-x}
\end{equation}
and $Z$ on $R(r_k,n-1)$ for every $k$.  Thus
$|M_R|=q+2$ and
\begin{equation}
 \operatorname{comp}(M_R)=(3,0,q-1).
\end{equation}

Consider the $q-1$ $X$-check supports
\begin{equation}
 \begin{split}
 D_k={}&\{L(r_k,n-1),L(r_k,0),L(r_k,1)\}\\
 &\cup\{R(r_k,n-1),R(r_k+1,n-1),R(r_k+2,n-1)\},
 \end{split}
 \label{eq:sm-hgp-long-detectors}
\end{equation}
and the following three $Z$-check supports:
\begin{align}
 A_0={}&\{L(n-1,0),L(0,0),L(1,0),
          R(1,n-2),R(1,n-1),R(1,0)\},\nonumber\\
 A_1={}&\{L(0,1),L(1,1),L(2,1),
          R(2,n-1),R(2,0),R(2,1)\},\nonumber\\
 A_2={}&\{L(1,n-1),L(2,n-1),L(3,n-1),
          R(3,n-3),R(3,n-2),R(3,n-1)\}.
 \label{eq:sm-hgp-end-detectors}
\end{align}
They are the raw BB check rows $(r_k,n-1)$ of $H_X$ and
$(1,0),(2,1),(3,n-1)$ of $H_Z$, respectively.  All $q+2$ physical supports
in Eqs.~\eqref{eq:sm-hgp-long-detectors} and
\eqref{eq:sm-hgp-end-detectors} are pairwise disjoint: the $R$ triples of
the $D_k$ partition the first-coordinate set
$\{4,\ldots,n-1,0\}$, their $L$ parts have distinct first coordinates,
and the $A_i$ use the complementary $R$ coordinates $1,2,3$ and distinct
$L$ second coordinates.  Each support has odd
syndrome on $M_R$: $D_k$ contains the corresponding core $Z$, while each
$A_i$ contains exactly one of the three $X$ sites in
Eq.~\eqref{eq:sm-hgp-mr-x}.  Every error with the same syndrome must
therefore occupy at least one qubit in every support, proving
\begin{equation}
 w_{\min}(s_\star)\ge q+2.
 \label{eq:sm-hgp-detector-lower}
\end{equation}
The reference error attains the bound, so equality holds.

At equality the packing is saturated: there is exactly one occupied qubit
in each $D_k$ and $A_i$, and no occupied qubit outside their union.  In
particular, the $Z$ component has at most one site in each detector cell.
Let $\Omega$ denote this union and let $C$ be the reference $Z$ vector on
the $q-1$ core sites $R(r_k,n-1)$.  We first solve the restricted equation
$H_Xz=H_XC$ exactly.  The kernel of $H_X|_{\Omega}$ has the five-vector
basis
\begin{equation}
 \mathbf 1_{A_0},\quad \mathbf 1_{A_1},\quad \mathbf 1_{A_2},
 \quad u_{01},\quad u_{12},
 \label{eq:sm-hgp-five-z-generators}
\end{equation}
where $\mathbf 1_{A_i}$ is the incidence vector of the $Z$-check support
$A_i$, and $u_{ab}$ occupies $R(j,n-1)$ for the two residue classes
$j\bmod3\in\{a,b\}$.  All five displayed vectors lie in the kernel and
are independent.  Completeness can be read directly from the raw rows.
For a homogeneous solution $w=z+C$, denote its two block components by
$\ell_{i,j}$ and $r_{i,j}$, with indices modulo $n$.  Every $H_X$ row is
\begin{equation}
 \ell_{i,j}+\ell_{i,j+1}+\ell_{i,j+2}
 +r_{i,j}+r_{i+1,j}+r_{i+2,j}=0.
 \label{eq:sm-hgp-raw-hx-recurrence}
\end{equation}
Restricted to $\Omega$, the layer $j=1$ sets every interior
$\ell_{r_k,1}$ to zero and gives
$\ell_{0,1}=\ell_{1,1}=\ell_{2,1}=r_{2,1}\equiv\alpha_1$.
The layer $j=-3$ similarly gives
$\ell_{1,-1}=\ell_{2,-1}=\ell_{3,-1}=r_{3,-3}\equiv\alpha_2$
and kills every interior $\ell_{r_k,-1}$.  Next, $j=0$ kills the
interior $\ell_{r_k,0}$ and forces
\begin{equation}
 r_{2,0}=\alpha_1,\qquad
 \ell_{-1,0}=\ell_{0,0}=\ell_{1,0}=r_{1,0}\equiv\alpha_0.
\end{equation}
Finally, the $j=-2$ layer gives $r_{1,-2}=\alpha_0$ and
$r_{3,-2}=\alpha_2$.  Subtracting
$\sum_i\alpha_i\mathbf 1_{A_i}$ now removes every $L$ coordinate and every
$R$ coordinate off the column $j=-1$.  The remaining $j=-1$ equations are
$r_i+r_{i+1}+r_{i+2}=0$, whose two-dimensional period-three kernel is
spanned by $u_{01}$ and $u_{12}$.  This proves completeness and gives
\begin{equation}
 \operatorname{rank}(H_X|_{\Omega})=6(q+2)-5=6q+7.
 \label{eq:sm-hgp-restricted-z-rank}
\end{equation}
Thus every candidate $Z$ component is uniquely
\begin{equation}
 C+\alpha_0\mathbf 1_{A_0}+\alpha_1\mathbf 1_{A_1}
 +\alpha_2\mathbf 1_{A_2}+\beta_{01}u_{01}+\beta_{12}u_{12}.
 \label{eq:sm-hgp-five-bit-z-form}
\end{equation}
The one-hot inequalities reduce these $2^5$ states to the following
four-case table, independent of $q$:
\begin{equation}
\begin{array}{c|c|c}
 (\beta_{01},\beta_{12}) & \text{admissible }(\alpha_0,\alpha_1,\alpha_2)
 & Z\text{ support}\\ \hline
 (0,0)&(0,0,0)&C\\
 (1,0)&(0,0,0)&C+u_{01}\\
 (0,1)&(0,0,0)&C+u_{12}\\
 (1,1)&\text{none}&C+u_{02}\text{ has three sites in a long cell}
\end{array}
 \label{eq:sm-hgp-five-bit-z-table}
\end{equation}
Every nonzero $\alpha_i$ likewise produces a detector cell with at least
two $Z$-bearing sites.  Hence all long cells are pure $Z$ on the $R$ column
$j=n-1$, and the only possible endpoint $Z$ sites are $R(i,n-1)$,
$i=1,2,3$.  The verification artifact emits the complete 32-state parity
table and checks the rank calculation directly from the raw BB rows.

For $C$, none of the endpoint detectors contains a $Z$ site.  For
$C+u_{01}$, the unique $Z$ sites of $A_0$ and $A_2$ are respectively
$R(1,n-1)$ and $R(3,n-1)$, while $A_1$ contains none; for $C+u_{12}$ the
corresponding occupied detectors are $A_0$ and $A_1$.  A detector already
occupied by $Z$ has two phase-free choices, $Z$ or $Y$, whereas an unoccupied
endpoint detector has six possible $X$ sites.  This gives the pre-parity
counts $6^3$, $2^2\cdot6$, and $2^2\cdot6$ below.

For completeness, write $x_Q\in\mathbb F_2$ for the $X$ component at an
endpoint qubit $Q$.  After the bulk pivots above, the raw $H_Z$ rows at
$(2,n-1),(3,0),(1,1),(3,1),(1,n-1)$ form the following complete reduced cap
system:
\begin{align}
 x_{L(1,n-1)}+x_{L(2,n-1)}+x_{R(2,n-1)}&=1,\nonumber\\
x_{L(1,0)}+x_{R(3,n-2)}+x_{R(3,n-1)}&=1,\nonumber\\
x_{L(0,1)}+x_{L(1,1)}+x_{R(1,0)}+x_{R(1,n-1)}&=1,\nonumber\\
x_{L(1,1)}+x_{L(2,1)}+x_{R(3,n-1)}&=1,\nonumber\\
x_{L(1,n-1)}+x_{R(1,n-2)}+x_{R(1,n-1)}&=1.
 \label{eq:sm-hgp-cap-parities}
\end{align}
Together with one-hot occupancy of $A_0,A_1,A_2$, these equations reduce the
$216$, $24$, and $24$ a priori endpoint choices for
$C,C+u_{01},C+u_{12}$ to respectively two, one, and one completions.  The
complete truth table is
\begin{equation}
\begin{array}{c|c|c|c}
 Z\text{ component} & \text{choices before parity} & X\text{ completion}
 & \text{composition}\\ \hline
 C & 6^3=216 & R(1,n-1),R(2,n-1),R(3,n-1) & (3,0,q-1)\\
 C & 6^3=216 & L(1,0),L(1,1),L(1,n-1) & (3,0,q-1)\\
 C+u_{01} & 2^2\!\cdot6=24 & R(1,n-1),R(2,n-1),R(3,n-1)
 & (1,2,q-1)\\
 C+u_{12} & 2^2\!\cdot6=24 & R(1,n-1),R(2,n-1),R(3,n-1)
 & (1,2,q-1)\\
 C+u_{02} & 0 & \text{excluded by long-cell one-hotness} & --
\end{array}
\label{eq:sm-hgp-four-state-table}
\end{equation}
In the third row the sites $R(1,n-1)$ and $R(3,n-1)$ are $Y$; in the
fourth row $R(1,n-1)$ and $R(2,n-1)$ are $Y$.  All other listed end sites
are $X$, and the long-column sites are $Z$.  Equations
\eqref{eq:sm-hgp-cap-parities} involve only coordinates
$0,1,n-2,n-1$ and are therefore independent of $q$ for every $n=3q\geq6$.
Direct substitution gives the same reduced system for the overlapping cases
$q=2,3$.  An exact row-by-row generator of the $216/24/24\to2/1/1$ reduction
is included in the verification artifact.

The first two representatives differ by the weight-six $X$ stabilizer at
the corner and therefore lie in one logical class $A$.  The last two are
obtained from the first representative by the two distinct nonzero
minimum-$Z$ logical lines $u_{01}$ and $u_{12}$, so they occupy distinct
classes $B$ and $C$.  The complete leading spectrum is consequently
\begin{equation}
 P_A=2x^3z^{q-1},\qquad
 P_B=P_C=xy^2z^{q-1}.
 \label{eq:sm-hgp-uniform-spectrum}
\end{equation}
For every positive Pauli direction satisfying
\begin{equation}
 x<y<\sqrt2\,x,
 \label{eq:sm-hgp-uniform-cone}
\end{equation}
each singleton in $B$ or $C$ has larger peak probability than an individual
$A$ representative, while the twofold mass of $A$ is larger than either
singleton mass.  MP and MLD therefore have disjoint winner sets.

For $q=1$, an exhaustive fixed-syndrome certificate proves the absence of
representatives below weight three and enumerates all 108 leading
representatives.  One logical class has six representatives of composition
$(3,0,0)$, while competing singleton classes have composition $(1,2,0)$.
Thus the same strict split holds throughout the common cone
Eq.~\eqref{eq:sm-hgp-uniform-cone}.  Combining this base certificate and the
all-$q\geq2$ gadget with
Eq.~\eqref{eq:sm-hgp-all-q-strengthened-lower} proves
\begin{smtheorem}[Exact onset of the separable HGP/BB family]
For every $q\ge1$, on the open interior Pauli cone
$0<x<y<\sqrt2x$ with $x+y+z=1$, the strict MP--MLD onset is
\begin{equation}
 m(\mathcal H_q)=q+2=\frac{d_q}{2}+2.
 \label{eq:sm-hgp-exact-onset}
\end{equation}
\end{smtheorem}

\section{Gross BB memory: exact one-order delay}\label{sec:sm-gross}

\subsection{Pauli-purity rigidity at even half distance}
\label{sec:sm-pauli-purity-rigidity}

Call a phase-free Pauli \emph{Pauli-pure} if all its occupied sites carry
the same letter $X$, $Y$, or $Z$.

\begin{smtheorem}[Pauli-purity rigidity]
\label{thm:sm-pauli-purity-rigidity}
Let $\mathcal Q$ be an even-distance stabilizer code of distance $d=2h$.
If every nontrivial centralizer of weight $d$ is Pauli-pure, then no strict
MP--MLD split occurs at physical weight $h$, for any interior memoryless
Pauli direction.  Consequently
\begin{equation}
 m(\mathcal Q)\geq h+1=\frac d2+1.
 \label{eq:sm-pauli-purity-delay}
\end{equation}
\end{smtheorem}

\begin{proof}
Fix a syndrome whose minimum fiber $\mathcal V_s$ has weight $h$ and
contains more than one logical class.  Join two representatives whenever
their logical classes differ.  As in the excess-two argument, the resulting
competition graph is complete multipartite and connected.

For every edge $E$--$F$, the product $U=E^\dagger F$ is a nontrivial
centralizer.  Hence $|U|\geq d$, while
$|U|\leq |E|+|F|=2h=d$; therefore $|U|=d$.  On
$\operatorname{supp}(U)$, let $c$ count sites occupied by both factors.
Outside $\operatorname{supp}(U)$, the identity product
forces the factors to agree; let $o$ count their common occupied sites.
The exact phase-free Pauli identity is
\begin{equation}
 |E|+|F|-|U|=c+2o.
 \label{eq:sm-zero-excess-identity}
\end{equation}
Its left-hand side vanishes, so $c=o=0$.  Thus $E$ and $F$ are disjoint
complementary halves of $\operatorname{supp}(U)$.  By hypothesis $U$ has one Pauli letter
$\sigma\in\{X,Y,Z\}$ on its entire support.  At a site occupied by only one
factor, that factor must also carry $\sigma$.  Both endpoints consequently
have the same composition: $h$ copies of $\sigma$.

Every edge of the connected competition graph preserves composition, so the
entire minimum fiber has a single composition.  All represented logical
classes are therefore MP-tied.  MLD may prefer classes of larger
multiplicity, but its winner set remains a nonempty subset of this MP tie
set.  Strict best-tie separation is impossible.  Together with the universal
half-distance obstruction this proves Eq.~\eqref{eq:sm-pauli-purity-delay}.
\end{proof}

\subsection{Certified Gross-code delay}
\label{sec:sm-gross-purity-certificate}

For the standard Gross memory of Ref.~\cite{BravyiEtAl2024Gross}
\begin{equation}
 \mathcal G=Q(x^3+y+y^2,y^3+x+x^2;12,6)=[[144,12,12]],
 \label{eq:sm-gross-code}
\end{equation}
the exact decision problem for a non-Pauli-pure weight-twelve nontrivial
centralizer was partitioned into the 24 disjoint sectors defined by the first
nonzero logical-label bit. In each sector the Boolean satisfiability (SAT)
instance imposes zero
syndrome, a nonzero quotient label, exact physical weight twelve, and the
presence of at least two of the three occupied-site categories $X$-only,
$Y$, and $Z$-only. All 24 sectors are exhaustively unsatisfiable (UNSAT).
Hence every
minimum Gross logical is Pauli-pure, and
Theorem~\ref{thm:sm-pauli-purity-rigidity} gives the unconditional bound
\begin{equation}
 m(\mathcal G)\geq7.
 \label{eq:sm-gross-certified-lower}
\end{equation}

As a separate regression certificate, translation closure followed by
seeded exhaustive completion yields all 3780 weight-twelve nontrivial
representatives: 1884 pure $X$, 12 pure $Y$, and 1884 pure $Z$, spanning 504
logical labels.  Their complete complementary-factor audit contains
1,452,000 multiclass weight-six syndromes and certifies every one either by
classwise singleton structure or by a common peak composition.  This finite
enumeration independently reproduces the analytic no-split conclusion at
half distance. The released JavaScript Object Notation (JSON) summaries record the shard terminations,
counts, and spectrum digest; the scripts reconstruct every Boolean instance
from the published Gross check matrices.

\subsection{Odd short-centralizer gap}
\label{sec:sm-gross-weight-thirteen-gap}

The Gross code has no nontrivial centralizer of weight thirteen.  The pure
$X$ and pure $Z$ sectors are direct exhaustive UNSAT instances.  The pure
$Y$ sector is partitioned by the 24 first-nonzero logical-label bits, and all
24 shards are exhaustive UNSAT.  For the remaining non-Pauli-pure sector, a
smaller exact symmetry cover suffices.  Such an operator contains at least
one $Y$ or $Z$-only site.  Translations act transitively on the $12\times6$
coordinates within each of the two Gross qubit blocks, so translating that
site to the block origin reduces the complete decision problem to four
anchors: left/right block times $Y/Z$.  All four anchored instances are
exhaustive UNSAT.  These mutually exhaustive sectors prove
\begin{equation}
 \nexists\,U\in N(\mathcal S)\setminus\mathcal S
 \quad\text{with}\quad |U|=13.
 \label{eq:sm-gross-no-weight-thirteen}
\end{equation}
Consequently every inter-class competition between weight-seven errors has
product weight twelve or fourteen.  The weight-seven decision is determined
by two syndrome factorization subgraphs: the excess-two graph of the
complete weight-twelve list and the complementary-factor graph of
weight-fourteen centralizers.

\subsection{Exact Gross onset from an excess-two star}
\label{sec:sm-gross-exact-onset}

Label the two Gross qubit blocks by $L(a,b)$ and $R(a,b)$, with
$a\in\mathbb Z_{12}$ and $b\in\mathbb Z_6$.  Define
\begin{align}
 E_A={}&Y_{L(0,4)}Y_{L(0,5)}
 X_{L(1,2)}X_{L(1,4)}X_{L(2,0)}X_{L(2,2)}X_{R(0,0)},
 \nonumber\\
 E_{B,1}={}&Z_{L(0,4)}Z_{L(0,5)}X_{L(3,1)}
 X_{R(1,0)}X_{R(1,1)}X_{R(1,2)}X_{R(2,1)},
 \nonumber\\
 E_{B,2}={}&Z_{L(0,4)}Z_{L(0,5)}X_{L(0,2)}X_{L(0,3)}
 X_{R(0,4)}X_{R(1,0)}X_{R(1,2)}.
 \label{eq:sm-gross-star-errors}
\end{align}

\begin{smtheorem}[Exact Gross onset]
\label{thm:sm-gross-exact-onset}
For the standard Gross code in Eq.~\eqref{eq:sm-gross-code}, throughout
\begin{equation}
 x>0,\qquad z>0,\qquad z<y<\sqrt2\,z,\qquad x+y+z=1,
 \label{eq:sm-gross-onset-cone}
\end{equation}
the strict best-tie onset is
\begin{equation}
 m(\mathcal G)=7=\frac d2+1.
 \label{eq:sm-gross-exact-onset}
\end{equation}
\end{smtheorem}

\begin{proof}
Equation~\eqref{eq:sm-gross-certified-lower} gives $m(\mathcal G)\ge7$.
Direct substitution of Eq.~\eqref{eq:sm-gross-star-errors} into the Gross
check matrix shows that all three errors have the same syndrome.  The two
$B$ errors have the same quotient label, while the $A$ error has a different
label.  Their compositions are
\begin{equation}
 \operatorname{comp}(E_A)=(5,2,0),\qquad
 \operatorname{comp}(E_{B,1})=
 \operatorname{comp}(E_{B,2})=(5,0,2).
 \label{eq:sm-gross-star-compositions}
\end{equation}

The fixed-syndrome certificate reconstructs all 132 Gross
parity equations directly from Eq.~\eqref{eq:sm-gross-code}.  The instance
with physical weight at most six is UNSAT.  The exact-weight-seven instance
terminates exhaustively after precisely the three operators in
Eq.~\eqref{eq:sm-gross-star-errors}.  Hence the complete leading logical-class spectrum is
\begin{equation}
 P_A=x^5y^2,\qquad P_B=2x^5z^2.
 \label{eq:sm-gross-star-spectrum}
\end{equation}
Condition $y>z$ makes the largest individual monomial occur uniquely in
class $A$, whereas $y^2<2z^2$ makes the total leading coset mass larger in
class $B$.  The two winner sets are therefore disjoint on the open cone
\eqref{eq:sm-gross-onset-cone}, proving the matching upper bound and the
theorem.
\end{proof}

For this syndrome the positive leading mass advantage that enters the
operational-gap formula of Sec.~\ref{sec:sm-gap} is
\begin{equation}
 P_B-P_A=x^5(2z^2-y^2)>0.
 \label{eq:sm-gross-star-margin}
\end{equation}
The released verifier uses the rational interior point
$(x,y,z)=(1/10,23/50,11/25)$ and records the complete spectrum, representative
supports, solver termination states, and Secure Hash Algorithm 256-bit
(SHA-256) digests of the source
certificates.

\section{Tie-safe operational failure gap}\label{sec:sm-gap}

Give MP its best possible tie rule.  Define
\begin{align}
 S_{\MP}^*(t)&=\sum_s\max_{L\in D_{\MP}(s,t)}M_{s,L}(t),\nonumber\\
 S_{\MLD}(t)&=\sum_s\max_LM_{s,L}(t),\\
 \Delta(t)&=P_{\rm fail}^{\MP,*}(t)-P_{\rm fail}^{\MLD}(t)
 =S_{\MLD}(t)-S_{\MP}^*(t).\nonumber
\end{align}
For each syndrome,
\begin{equation}
 \delta_s(t)=\max_LM_{s,L}(t)
 -\max_{L\in D_{\MP}(s,t)}M_{s,L}(t)\geq0.
\end{equation}
The constrained and unrestricted maxima agree exactly when
$D_{\MP}\cap D_{\MLD}\neq\varnothing$.  Therefore
\begin{equation}
 \delta_s>0\quad\Longleftrightarrow\quad
 D_{\MP}(s,t)\cap D_{\MLD}(s,t)=\varnothing,
\qquad \Delta=\sum_s\delta_s.
\end{equation}
No cancellation between syndromes is possible.  The explicit strict
witness in each family is a leading strict split in the sense of
Theorem~\ref{thm:sm-spectrum}.  Put
\begin{equation}
 {\cal A}_s=\argmax_LT_{s,L},\qquad
 \kappa_s=\max_LP_{s,L}
 -\max_{L\in{\cal A}_s}P_{s,L}.
 \label{eq:sm-kappa}
\end{equation}
Taking the leading term in Eq.~(\ref{eq:sm-exact-expansion}) gives
\begin{equation}
 \delta_s(t)=\kappa_st^{w_s}
 +O(t^{w_s+1}).
\end{equation}
Therefore, if no syndrome with $w_s<m$ exhibits strict disagreement and
$m=\min\{w_s:\kappa_s>0\}$, as established for the six families
studied here, then
\begin{equation}
 \Delta(t)=K_mt^m+O(t^{m+1}),\qquad
 K_m=\sum_{s:w_s=m}\kappa_s>0.
 \label{eq:sm-exact-K}
\end{equation}
Equation~\eqref{eq:sm-exact-K} identifies the leading coefficient as a sum
of the positive syndrome margins at weight $m$.
The lower-order exclusions and strict witnesses proved in Secs.~\ref{sec:sm-five} through \ref{sec:sm-gross}
therefore give
\begin{align}
 \Delta_\ell^{(4)}(t)
 &=A_\ell t^{\lceil d_\ell/2\rceil}
 +O(t^{\lceil d_\ell/2\rceil+1}),\nonumber\\
 \Delta_\ell^{(5)}(t)
 &=B_\ell t^{(3^\ell+1)/2}
 +O(t^{(3^\ell+3)/2}),\nonumber\\
 \Delta_d^{\rm planar}(t)
 &=C_dt^{(d+1)/2}+O(t^{(d+3)/2}),\nonumber\\
 \Delta_d^{\rm toric}(t)
 &=D_dt^{(d+3)/2}+O(t^{(d+5)/2}),\nonumber\\
 \Delta_q^{\rm HGP}(t)
 &=E_qt^{q+2}+O(t^{q+3}),\nonumber\\
 \Delta^{\rm Gross}(t)
 &=Ft^7+O(t^8),
\end{align}
with $A_\ell,B_\ell,C_d,D_d,E_q,F>0$ on the corresponding open regions.

For the recurrent five-qubit witness, the coefficient ratio in
Eq.~(\ref{eq:sm-five-ratios}) gives the normalized margin
\begin{equation}
 \eta_5=1-\frac{x^2+y^2}{2yz}.
\end{equation}
For the toric gadget at $d=2\rho+1\geq5$ and the planar boundary path, respectively,
\begin{align}
 \kappa_s^{\rm toric}
 &=x^{\rho-1}(2z^3-y^3),&
 \eta_{\rm toric}
 &=1-\frac{y^3}{2z^3},\nonumber\\
 \eta_{\rm planar}
 &=1-\frac{P_A}{P_B}
 >1-\frac{y}{2z(1-x^2/y^2)},\nonumber\\
 \eta_{\rm HGP}
 &=1-\frac{y^2}{2x^2}\quad(q\geq2),&
 \eta_{\rm Gross}
 &=1-\frac{y^2}{2z^2}.
\end{align}
For the $q=1$ HGP base case, whose mass-winning class has six rather
than two representatives, $\eta_{\rm HGP}=1-y^2/(6x^2)$.
Here $\eta_s=\kappa_s/\max_LP_{s,L}$.  Within each stated family regime,
with the $q=1$ HGP base case treated separately, these normalized margins
are independent of depth or distance.  Hence, on compact subsets away from
a decision boundary, they have a uniform positive lower bound.

\section{Independent finite-noise benchmark and distance scaling}\label{sec:sm-benchmark}

Equation~(\ref{eq:sm-exact-expansion}) first yields an explicit rigorous
finite-noise window. Because $x+y+z=1$, the total composition weight of all
weight-$k$ Pauli strings is $\binom nk$. Thus every higher-weight tail in one
coset obeys
\begin{equation}
 0\leq\sum_{j=1}^{n-w}\xi^jP_{s,L}^{(w+j)}
 \leq B_{n,w}(\xi):=
 \sum_{j=1}^{n-w}\binom n{w+j}\xi^j.
 \label{eq:sm-tail-bound}
\end{equation}
If $B_{n,w_s}(\xi)<\kappa_s$ and
$\xi<\max_LT_{s,L}$, the leading MLD class remains strictly above every
MP-allowed class and no higher-weight microscopic error changes the MP set.
This gives a completely explicit sufficient window. Replacing
$B_{n,w}$ by the syndrome-resolved enumerator gives the sharp computable
crossover tested below.

An independent Stim calculation reconstructs one explicit witness
syndrome in each matched $d=3$ planar and toric memory from
Pauli multiplication and commutation operations \cite{Gidney2021Stim}.
The complete logical-coset orbit is then summed
exactly at finite $t$ along
$(x,y,z)=(1/10,46/100,44/100)$.  The resulting single-syndrome leading
coefficients are
\begin{align}
 \kappa_s^{\rm planar,\ d=3}&=\frac{21}{250},
 &m_s^{\rm planar}&=2,\nonumber\\
 \kappa_s^{\rm toric,\ d=3}&=\frac{21}{5000},
 &m_s^{\rm toric}&=3.
\end{align}
The planar calculation exhausts its $4^9$ Pauli errors, whereas the
toric calculation exhausts the fixed syndrome stabilizer orbit of
$2^{n-k}2^{2k}=2^{20}$ elements. Figure~3 of the main article
compares the exact gap with $\kappa_st^{m_s}$ over
$4\times10^{-4}\leq t\leq0.05$.  At $t=0.03$, the exact normalized
ratios are $0.855737\ldots$ and $0.726495\ldots$; the leading term stays
within ten percent through approximately $t=0.0204$ and $t=0.00990$ for
the planar and toric witnesses, respectively.

To quantify the finite-$t$ window across distance, we next contract the
same analytic witness fibers.
Each physical qubit supplies a local factor over the incident independent
stabilizer-generator bits.  Sum-product elimination gives the exact
logical-coset mass, while replacing each sum by a maximum gives the exact
MP peak on the identical factor graph.  At $d=3$, both contractions agree
pointwise with the independent Stim orbit sums above.  We then evaluate
planar distances $d=3,5,7,9$ and toric distances $d=3,5,7,9$. For toric
$d=9$, an exact sliced contraction evaluates $4096$ slices of the same local
factor graph.

For each distance, let $\delta_{s,d}(t)$ denote the fixed-witness gap
$\delta_s(t)$ for the distance-$d$ member and let $\kappa_d,m_d$ be its leading
coefficient and onset exponent. Define
\begin{equation}
 R_d(t)=\frac{\delta_{s,d}(t)}{\kappa_dt^{m_d}},\qquad
 t_\epsilon(d)=\sup\{t:|R_d(u)-1|\leq\epsilon
 \ \text{for all }0<u\leq t\}.
 \label{eq:sm-distance-window}
\end{equation}
Here $t_\epsilon(d)$ measures the relative accuracy of the leading
fixed-syndrome asymptotic formula. Its contraction with distance reflects
the growing volume factor from qubits with no error, while the gap itself
remains suppressed by $t^{m_d}$ with increasing $m_d$.
The raw leading approximation omits the identity probability carried by
every representative of minimum weight. Retaining that factor exactly gives the
parameter-free LO+I approximation in Fig.~6(c),
\begin{equation}
 \delta_{s,d}^{(0,\mathrm{id})}(t)
 =\kappa_dt^{m_d}(1-t)^{n_d-m_d}.
 \label{eq:sm-lo-plus-i}
\end{equation}
Equation~\eqref{eq:sm-lo-plus-i} retains the exact contribution of the
imbalance at minimum weight with the factors for no error unexpanded. It
separates the elementary volume correction, of scale $(n_d-m_d)t$, from
genuine corrections from the coset spectrum at higher weights.

After this identity factor is removed, write the exact gap for a fixed syndrome as
\begin{equation}
 \delta_{s,d}(t)
 =\kappa_dt^{m_d}(1-t)^{n_d-m_d}
 \left[1+c_dt+O(t^2)\right].
 \label{eq:sm-lo-plus-i-correction}
\end{equation}
The coefficient $c_d$ is the normalized net contribution of the next weight
layer and depends on the multiplicities and Pauli compositions of the selected
witness. Thus the LO validity window
contracts clearly with distance, while LO+I removes the dominant volume
dependence and leaves a window controlled by the discrete local spectrum at
higher weight. Its distance dependence is therefore determined by that
discrete spectrum and can be nonmonotone.

Along the displayed ray, the raw ten percent windows for planar
$d=3,5,7,9$ are respectively
 $0.0204$, $0.00618$, $0.00247$, and $0.00151$; for toric $d=3,5,7,9$ they are
 $0.00990$, $0.00270$, $0.00116$, and $0.000659$. The LO+I values increase to
 $0.0497$, $0.0192$, $0.0380$, and $0.0148$ for the planar witnesses and
 $0.0211$, $0.0139$, $0.0346$, and $0.0330$ for the toric witnesses. The alternating planar
values reflect the checkerboard parity between two rows already present in the
boundary-path spectra. The nonmonotone LO+I windows measure the residual spectrum
at higher weight after removal of the common identity factor. The
toric $d=7$ and $d=9$ values are close because their first residual corrections
are nearly equal. The slightly smaller $d=9$ window is set by subsequent terms
in the local spectrum.
 Thus the bare $\kappa_dt^{m_d}$ window contracts
roughly with the code area, while Eq.~(\ref{eq:sm-lo-plus-i})
removes the dominant trivial volume dependence and provides a more useful
estimate at finite noise.

Figure~6 of the main article displays the finite noise curves, distance
ratios, and ten percent validity windows. The distance panels use exact sum
product and max product contractions on factor graphs, checked against the
independent $d=3$ orbit sums.

\FloatBarrier

\section{Exact verification boundary and released artifact}\label{sec:sm-artifact}

The released computational artifact verifies the finite statements and
certificate hypotheses entering the analytic results above.
The canonical certificate-verification commands perform the following
checks:
\begin{enumerate}
\item reconstructs the five-qubit entrance from the seed stabilizer;
\item checks the two-phase ten-state support and logical-label closure;
\item verifies 144 unique exponent branches, 36 coefficient returns,
factors all 19 rational certificate inequalities, and derives the analytic
region in Eq.~(\ref{eq:sm-five-region}) from six primitive signs;
\item reconstructs the four-qubit fixed witnesses through level 11 on a
$2\times2$ cover, verifies the level-12 entrance on a $4\times2$
log-monotonicity cover, and checks eight coefficient plus 278 exponent
Farkas rows;
\item recomputes the planar boundary-path spectra for the recorded finite
regression distances;
\item performs exact all-logical-class toric minimization at
$d=3,5,7,9,11$ and compares it with
Eq.~(\ref{eq:sm-toric-spectra});
\item verifies the HGP/BB four-step seed lemmas, detector packing, reduced
five-variable fiber, and endpoint-cap table underlying
$m(\mathcal H_q)=q+2=d/2+2$;
\item reconstructs the Gross excess-two star, exhausts its fixed syndrome
weight-seven fiber, and validates the complete half-distance purity summary;
\item reconstructs the complete small-code composition spectra and the
finite $t$ operational gap data in Fig.~6(a) of the main article.
\end{enumerate}
\emph{Reproducibility.}
The versioned verification artifact provides exact, machine-checkable
reproductions of all finite certificates used in the proofs. Running
\texttt{python verify.py --quick} validates the archived certificates and
replays the inexpensive HGP/BB and Gross checks, while
\texttt{python verify.py} additionally regenerates the invariant-cone,
Farkas, and toric certificates. The documented optional workflows reproduce
the complete HGP/BB and 24-sector Gross audits, as well as the independent
finite $t$ benchmarks in Fig.~6 of the main article. The theorem-critical
checks use integer, rational, Boolean, symbolic, or rigorously
outward-rounded interval arithmetic. The finite-$t$ curves are plotted from
the fixed exact spectra. The finite topological computations provide
regression checks for the analytic all-distance proofs. Full commands,
certificate hashes, expected outputs, and software requirements are
documented in the artifact README~\cite{ZhaoYanZenodo2026}.


\end{document}